\documentclass{article}
\usepackage[utf8]{inputenc}
\usepackage{authblk}
\usepackage{floatrow}
\usepackage[usenames,dvipsnames]{xcolor} 
\usepackage{amsmath}
\usepackage[utf8]{inputenc}
\usepackage{amsfonts}
\usepackage{amssymb}
\usepackage{tikz-cd}
\usepackage{graphicx}
\usepackage{bm}
\usepackage[numbers,sort]{natbib}
\usepackage{enumitem}
\usepackage[margin=1.33in]{geometry}
\usepackage[pdftex,
pdfpagemode=UseNone,
breaklinks=true,
extension=pdf,
colorlinks=true,
linkcolor=blue,
citecolor=PineGreen,
urlcolor=blue,
]{hyperref}
\usepackage{amsthm}
\usepackage{algorithm}
\usepackage{algorithmic}

\theoremstyle{theorem}
\newtheorem{theorem}{Theorem}[section]

\newtheorem{lemma}[theorem]{Lemma}
\newtheorem{corollary}[theorem]{Corollary}
\newtheorem{proposition}[theorem]{Proposition}

\theoremstyle{definition}

\newtheorem{definition}[theorem]{Definition}
\newtheorem{examplex}[theorem]{Example}

\newtheorem*{problem}{Problem formulation}
\newenvironment{example}
{\pushQED{\qed}\examplex}
{\popQED\endexamplex}

\newtheorem{remarkx}[theorem]{Remark}
\newenvironment{remark}
{\pushQED{\qed}\remarkx}
{\popQED\endremarkx}

\DeclareMathOperator{\NN}{\mathbb{N}}
\DeclareMathOperator{\CC}{\mathbb{C}}
\DeclareMathOperator{\RR}{\mathbb{R}}

\DeclareMathOperator{\QQ}{\mathbb{Q}}
\DeclareMathOperator{\KK}{\mathbb{K}}
\DeclareMathOperator{\LL}{\mathbb{L}}
\newcommand{\Va}{\mathbb{V}}
\newcommand{\Cu}{\mathcal{C}}
\newcommand{\LD}{\mathcal{L}}

\DeclareMathOperator{\denom}{denom}
\DeclareMathOperator{\num}{num}

\usepackage{listings}

\definecolor{dkgreen}{rgb}{0,0.6,0}
\definecolor{gray}{rgb}{0.5,0.5,0.5}
\definecolor{mauve}{rgb}{0.58,0,0.82}

\title{On real and observable rational realizations of input-output equations}

\author[1]{Sebastian Falkensteiner\thanks{sebastian.falkensteiner@mis.mpg.de}}
\affil[1]{Max Planck Institute for Mathematics in the Sciences, Inselstr. 22, 04103 Leipzig, Germany.}

\author[1,2]{Dmitrii Pavlov\thanks{dmitrii.pavlov@mis.mpg.de  (corresponding author)}}
\affil[2]{Technische Universität Dresden, Zellescher Weg 12--14, 01069 Dresden, Germany.}

\author[3]{J.Rafael Sendra\thanks{jrafael.sendra@cunef.edu}}
\affil[3]{CUNEF University, Department of Mathematics, Spain.}

\date{}

\begin{document}

\maketitle
{\color{red} The journal version of this paper appears in
\textit{[S. Falkensteiner, Dmitrii Pavlov, Rafael Sendra. \\
On real and observable rational realizations of input-output
equations \ 
Systems \& Control Letters 198 (2025) 106059 
\url{https://doi.org/10.1016/j.sysconle.2025.106059}]} under the CC BY license ( \url{http://creativecommons.org/licenses/by/4.0/} ).}
\begin{abstract}
Given a single (differential-algebraic) input-output equation, we present a method for finding different representations of the associated system in the form of rational realizations; these are dynamical systems with rational right-hand sides. 
It has been shown that in the case where the input-output equation is of order one, rational realizations can be computed, if they exist. 
In this work, we focus first on the existence and actual computation of the so-called observable rational realizations, and secondly on rational realizations with real coefficients. 
The study of observable realizations allows to find every rational realization of a given first order input-output equation, and the necessary field extensions in this process. 
We show that for first order input-output equations the existence of a rational realization is equivalent to the existence of an observable rational realization. 
Moreover, we give a criterion to decide the existence of real rational realizations. 
The computation of observable and real realizations of first order input-output equations is fully algorithmic. We also present partial results for the case of higher order input-output equations.
\end{abstract}

\noindent \textbf{keywords}
algebraic differential equations, rational dynamical systems, real realizations, observability, algebraic curves, proper parametrization

\section{Introduction}

Many processes in natural sciences are conveniently described by dynamical systems in the \emph{state space form}, that is, ODE systems of the form 
 \begin{equation}\label{eq:ODEmodel}
   \Sigma = \begin{cases}
    \mathbf{x}' = \mathbf{f}(\mathbf{u}, \mathbf{x}),\\
    \mathbf{y} = \mathbf{g}(\mathbf{u}, \mathbf{x}),
    \end{cases}
   \end{equation}
where $\mathbf{x} = (x_1, \ldots, x_n)$ are the unknowns describing the state of the system (\emph{state variables}), $\mathbf{u} = (u_1, \ldots, u_m)$ are the unknowns representing external forces (\emph{input variables}), $\mathbf{y} = (y_1,\ldots,y_l)$ are \emph{output variables}, $\mathbf{f} = (f_1, \ldots, f_n)$ are functions describing how the rate of change of the state depends on the state and external inputs and $\mathbf{g} = (g_1,\ldots,g_l)$ are functions describing how the output of the system depends on its state and inputs.

In an experimental setup it is typically only possible to observe the values of the input and output variables, but not of the state variables.
Therefore, by using numerical techniques, one might be able to find differential equations that connect $\mathbf{y}$ and $\mathbf{u}$ but not the variables $\mathbf{x}$. Such equations are called \emph{input-output equations} (IO-equations).
We note that in the case of rational single-output systems, i.e. when $l=1$ and $\mathbf{f}$ and $\mathbf{g}$ are tuples of rational functions, it is possible to find a single algebraic IO-equation describing all the relations between inputs and outputs of the system~\cite{dong2021differential}.

The question of reconstructing a dynamical system in the state space form from given input-output data (in the form on an input-output map, a set of input-output trajectories, or an input-output equation) is known as the \emph{realization problem} and is widely studied in control theory~\cite{sontag1993rational, schaft1986classic, sussmann1976classic, kotta2018minimal, jakubczyk1980minimal, zhang2010multi}. 
The structure of the problem depends significantly on which class of functions~$\mathbf{f}$ and~$\mathbf{g}$ are sought in. Typical classes considered include polynomial, rational, algebraic and analytic functions~\cite{sontag1993rational, schaft1986classic, sussmann1976classic}.
In this paper we concentrate on the version of the realization problem in which the starting point is the input-output equation. This is different from the articles cited above, where one starts either with an input-output map or a set of input-output trajectories. Our choice of formulation of the realization problem will allow us to use algorithmic tools to tackle it. Moreover, we concentrate on the case of single-output systems with~$\mathbf{f}$ and~$\mathbf{g}$ rational, since it is of substantial algebraic interest. 
The problem of recovering a system of the form~\eqref{eq:ODEmodel} with~$\mathbf{f}$ and $\mathbf{g}$ rational is known as the \emph{rational realization problem} \cite{bartosiewicz1985new, bartosiewicz1986realizations, moog1990prime, nvemcova2009realization, nvemcova2010realization, nvemcova2010structural}.

The realization problem was originally studied using mainly analytic methods~\cite{sussmann1976classic, jakubczyk1980minimal, schaft1986classic}.
The first attempt to approach it from the point of view of applied algebraic geometry was made by Forsman in~\cite{forsman1993rational}. 
He showed that for a no-input-single-output system over~$\mathbb{C}$ rational realizability is equivalent to unirationality of the hypersurface defined by a given IO-equation.
This approach was extended to the case of single-input-single-output systems over algebraically closed fields of zero characteristic in~\cite{pavlov2022realizing}.
However, from an applied point of view it is more interesting to consider the rational realization problem over the field of real numbers $\mathbb{R}$. 
In this paper we extend the setup of \cite{pavlov2022realizing} and study the rational realization problem over~$\mathbb{R}$. 

Just like in~\cite{forsman1993rational} and~\cite{pavlov2022realizing}, we show that rational realizations can be obtained from special parametrizations of the hypersurface defined by the IO-equation.
To decide whether real realizations exist and produce them algorithmically, it is crucial that the parametrization we start from is \emph{proper}, that is, induces a birational map from the affine space to the hypersurface. 
In control theory, this property is also called \emph{(global) observability}. 
This motivates studying the problem of finding \emph{observable rational realizations}, that is, deciding whether a realization can be obtained from a proper parametrization of the corresponding hypersurface. 
We also address this problem for the case of single-input-single-output systems. We note, however, that in many cases one can find real rational realizations of IO-equations without relying on proper parametrizations. It is exactly the desire to develop \textit{necessary and sufficient} conditions of rational realizability over $\mathbb{R}$ that motivates our study of global observability.



We now comment on the novelty of the paper. The problem of finding observable rational realizations has been studied in \cite{nvemcova2009realization, nvemcova2010realization}. An analog of a theorem by Sussmann \cite{sussmann1976classic} for rational systems \cite{nvemcova2010realization} states that if a rational realization exists, it can always be made globally observable if one allows it to be defined on a \emph{subvariety} of the affine space. In \cite{pavlov2022realizing} the authors study realizations that are necessarily defined on the \emph{whole} affine space but are a priori only \emph{locally} observable. In this paper, extending the approach of \cite{pavlov2022realizing}, we show (Theorems \ref{thm:properrealization2} and \ref{thm:firstorderproper}) that for IO-equations of low order, if a rational realization (defined on the whole affine space) exists, then a \emph{globally} observable rational realization defined on the \emph{whole} affine space also exists. Another novelty is the algorithmic approach to finding real rational realizations. While the problem of existence of real rational realizations was studied in \cite{nvemcova2009realization}, the authors of that paper underscore that the algorithms for producing such realizations with desirable properties still need to be developed. In this paper, we fill this gap for lower order IO-equations with Algorithms \ref{alg-properrealization} and \ref{alg-realrealization}. With this in mind, we note that studying other important control-theoretic properties of realizations, such as controllability or minimality, is outside the scope of this paper.

Many models studied in control theory involve parameters. One frequently studied property is that of \emph{identifiability} of the parameters~\cite{hong2020global}. 
This property is strongly related to observability, since a parameter $c$ can always be viewed as additional state variable when adding the additional equation $c'=0$. Then the parameter $c$ is globally identifiable if and only if it is, viewed as a state variable, globally observable. 
In~\cite{barreiro2023origins} one can find a summary and collection of examples on locally but not globally identifiable parameters. We use the relation between identifiability and observability to study the latter by using algebraic methods, see Section \ref{sec-proper}.

This paper is organized as follows.
In Section \ref{sec-prelim} we introduce our algebraic framework, give necessary definitions, recall several general results on rational realizability over subfields of~$\mathbb{C}$ from~\cite{pavlov2022realizing} and prove some new ones. 
In Section \ref{sec-proper} we study the problem of (globally) observable rational realizability. 
In particular, we show that for first order IO-equations realizability is equivalent to observable realizability and present an algorithm for finding observable realizations. 
For equations of order zero w.r.t. $u$ we show that observable realizability can be ensured by the properness of an intermediate realization in~\cite[Algorithm 1]{pavlov2022realizing}.
Finally, Section 4 is devoted to the problem of real realizability. We present a criterion for real realizability of an IO-equation with real coefficients and, for first-order IO-equations, an algorithm for finding real realizations. 

\section{Preliminaries} \label{sec-prelim}

\subsection{Basics on Algebraic Geometry and Differential Algebra}
We start by introducing some basic concepts of algebraic geometry and differential algebra that we will need throughout the paper. Some useful introductory references here are \cite{cox1997ideals} and \cite{pogudin2024differential}.

\begin{definition}
A \emph{polynomial} over a field $\mathbb{K}$ in the variables $x_1,\ldots,x_n$ is a finite linear combination of finite products of variables with coefficients in $\mathbb{K}$. The set of all such polynomials forms a ring which is denoted by $\mathbb{K}[x_1,\ldots,x_n]$. Note that polynomials in a countably infinite set of variables are defined in exactly the same way. A polynomial is called \emph{irreducible} if it cannot be written as a product of two nonconstant polynomials. A \emph{rational function} over $\mathbb{K}$ in $x_1,\ldots,x_n$ is an equivalence class of ratios of two polynomials under the standard equivalence $p/q\equiv r/s\iff ps-qr=0$, where $q$ and $s$ are assumed not to be zero polynomials. The set of all such rational functions forms a field which is denoted by $\mathbb{K}(x_1,\ldots,x_n)$. This is the field of fractions of $\mathbb{K}[x_1,\ldots,x_n]$. We say that a rational function is \emph{in reduced form} if it is represented by a ratio of coprime polynomials. Note that $\mathbb{K}(x_1)(x_2)=\mathbb{K}(x_1,x_2)$.
\end{definition}
 
\begin{definition}
    A \emph{differential ring} is a pair $(R, \delta)$, where $R$ is a commutative ring with identity and $\delta: R\to R$ is a \emph{derivation}, i.e. a linear map satisfying the Leibniz rule. That is, $\delta(a+b) = \delta(a) + \delta(b)$ and $\delta(ab) = a\delta(b) + b\delta(a)$ for all $a,b\in R$. For the sake of brevity we adopt the notation $a' := \delta(a)$.
\end{definition}

\begin{definition}
    Let $(R,\delta)$ be a differential ring. We write $x^{(\infty)} = (x, x', x'',x^{(3)},\ldots)$ for a countable set of formal variables. We equip the polynomial ring $R[x^{(\infty)}]$ with a derivation obtained by extending $\delta$ to the variables as $\delta(x^{(i)}) = x^{(i+1)}$ and then to the whole polynomial ring by linearity and the Leibniz rule. The resulting differential ring is called the  \emph{ring of differential polynomials} in the (differential) variable $x$. When writing $R[x^{(\infty)}]$ we will always assume that it is equipped with a derivation described above. If the derivation on $R$ is not specified, we will assume it is identically zero. In this case $R$ is a \emph{ring of constants}. We call the highest order of occurring derivative of $x$ in a polynomial $p\in R[x^{(\infty)}]$ the \textit{order} of $p$ w.r.t. $x$.
\end{definition}
\begin{definition}
    If $R=\mathbb{K}$ is a field, then the field of fractions of $\mathbb{K}[x^{(\infty)}]$ is denoted $\mathbb{K}(x^{(\infty)})$ and is called the \emph{field of differential rational functions} in $x$. We note that differential polynomials and rational functions in multiple differential variables are naturally defined by adjoining one differential variable at a time to the base ring or field.
\end{definition}

\begin{definition}
   Let $(R, \delta)$ be a differential ring. An \emph{ideal} in $R$ is a subset $I\subseteq R$ such that for any $a,b \in I$ and $c\in R$ we have $a+b\in I$ and $ac\in I$. A \emph{differential ideal} is an ideal that is closed under taking derivatives, i.e. $\delta(I)\subseteq I$. We say that a (differential) ideal $I$ is \emph{generated} by the elements $f_1,\ldots,f_m\in R$ if it is the inclusion-minimal (differential) ideal containing $f_1,\ldots,f_m$. We write $\langle f_1,\ldots, f_m \rangle$ for the ideal and $\langle f_1,\ldots,f_m \rangle^{(\infty)}$ for the differential ideal generated by $f_1,\ldots,f_m$. We say that an ideal $I \subseteq R$ is \emph{prime} if $ab\in I$ for $a,b\in R$ implies $a\in I$ or $b\in I$.
\end{definition}

\begin{definition}
    Let $\mathbb{K}$ be an algebraically closed field. An \emph{algebraic set} or an \emph{affine variety} in the linear space $\mathbb{K}^n$ is the common zero set of finitely many polynomials in $\mathbb{K}[x_1,\ldots,x_n]$. If the variety is defined by the vanishing of a single polynomial, we call it a \emph{hypersurface}.
\end{definition}

We denote by $\QQ$ the field of rational numbers, by $\RR$ the field of real numbers and by $\CC$ that of complex numbers. Let $\KK$ be a field such that $\QQ \subseteq \KK \subseteq \CC$ and let $\LL$ be a field extension of $\KK$.  Let $\KK[u_1,\ldots,u_m]$ be polynomials in new indeterminates $u_i$.
In what follows, we choose $\LL$ to be the field of rational functions $\KK(u_1,\ldots,u_m)$.

\begin{definition}
Let $F \in \LL[y_0,\ldots,y_n]$. 
Associated to~$F$, we denote by $\Va(F)$ the algebraic set $$\Va(F)=\{(b_0,\ldots,b_n)\in \overline{\LL}^{n+1} \mid F(b_0,\ldots,b_n)=0\},$$
where $\overline{\LL}$ is the algebraic closure of $\LL$. We call $\Va(F)$ the \textit{corresponding hypersurface} of $F$ (over $\overline{\LL}$).
\end{definition}

For brevity in what follows we will denote tuples of variables or functions by bold letters. 

\begin{definition}
    A \textit{(uni)rational parametrization} of $\Va(F)$ over a field $\mathbb{F}$, where $\LL \subseteq \mathbb{F} \subseteq \overline{\LL}$, is a tuple of rational functions $\mathbf{P}(\mathbf{x}) = (P_0(\mathbf{x}),\ldots,P_n(\mathbf{x})) \in \mathbb{F}(\mathbf{x})^{n+1}$ such that $F(\mathbf{P}(\mathbf{x}))=0$ and the Jacobian matrix of the map of vector spaces defined by $\mathbf{P}(\mathbf{x})$ has rank $n$ at almost every point. 
If such a rational parametrization exists, we say that $\Va(F)$ is \textit{parametrizable} or \textit{unirational} (over $\mathbb{F}$). 
Moreover, if $\mathbf{P}$ admits a rational inverse almost everywhere (i.e. everywhere except maybe a proper algebraic subset of $\Va(F)$) and thus the equality of fields $\mathbb{F}(\mathbf{P}(\mathbf{x}))=\mathbb{F}(\mathbf{x})$ holds, we call $\mathbf{P}$ \textit{proper} or \textit{birational} and $\Va(F)$ \textit{rational}. More details on birational maps are available in the introductory reference~\cite{shafarevich1994basic}.
\end{definition}

If $n \in \{1,2\}$, $\Va(F)$ is unirational over $\overline{\LL}$ if and only if it is rational~\cite{sendra2008rational,schicho1998rational}. In this case, rational parametrizations of $\Va(F)$ can be computed algorithmically (see Remark~\ref{rem-param}). 
Let us define the \emph{degree} of a tuple of rational functions $\mathbf{s}(\mathbf{x}) \in \overline{\LL}(\mathbf{x})^{n}$, where each component has coprime denominator and numerator, to be the maximum of the degrees of all numerators and denominators.
If $\mathbf{P}(\mathbf{x}) \in \overline{\LL}(\mathbf{x})^{n+1}$ is a proper parametrization of $\Va(F)$, then all other proper parametrizations of $\Va(F)$ are related by reparametrizations $\mathbf{P}(\mathbf{s}(\mathbf{x}))$, where $\mathbf{s}(\mathbf{x}) \in \overline{\LL}(\mathbf{x})^{n}$ is a birational transformation~\cite[Lemma 3.1]{arrondo1997parametric}. When $n=1$, this is equivalent to requiring that $\mathbf{s}$ has degree one and thus is a M\"obius transformation.

\begin{remark}
We have presented two ways of defining the hypersurface $\mathbb{V}(F)$: by giving its defining equation and by giving its parametrization. The problem of recovering the defining equation from a given parametrization of $\mathbb{V}(F)$ is known at the \emph{implicitization problem}~\cite[Section 3.3]{cox1997ideals}.
\end{remark}

Before we move on to defining the realization problem, we make two general comments about our setup. Firstly, we concentrate on realizations of systems represented by IO-equations. This approach is pursued for instance in \cite{sontag1993rational} but is different from e.g. \cite{nvemcova2009realization, nvemcova2010realization} where the system is represented by a response map rather than by an IO-equation. We chose the IO-equations approach because it allows us to use constructive methods of computational and differential algebra and provide concrete algorithms for finding realizations. Secondly, we will treat the rational realization problem in a purely algebraic way, by treating the input, state and output variables as formal variables that can be differentiated. This can be related to the framework of analytic functions by using Seidenberg's embedding theorem \cite{Seid1958} and Ritt's theorem of zeros \cite[p. 176]{Ritt}. Thus, $\mathbf{u}$ and $\mathbf{x}$ in our computations (although being formal variables) \emph{can be thought of} as analytic functions of time. We make this somewhat more precise in the following sections.

\subsection{Input-Output Equations and Rational Realizations}
Let $x_1,\ldots,x_n,y,u_1,\ldots,u_m$ be differential unknowns depending on a differential indeterminate $t$ and $x_i'=\frac{d}{dt}x_i$ denote the usual derivative w.r.t. $t$. In what follows, we treat our unknowns as formal variables and $\frac{d}{dt}$ as a formal derivation on the corresponding field of differential rational functions $\mathbb{K}(\mathbf{u}^{(\infty)}, \mathbf{x}^{(\infty)}, y^{(\infty)})$.

\begin{definition}
    By a \emph{rational system in state space form} we mean a system of ODEs in the formal variables $\mathbf{x}, \mathbf{u}$ and $y$ of the following form:

\begin{equation}\label{eq-realization}
\Sigma = \begin{cases}
    x_1' = p_1(\mathbf{u},\mathbf{x}),\\
   \,\,\,\,  \vdots  \\
x_n'=p_n(\mathbf{u},\mathbf{x}),\\
y=q(\mathbf{u},\mathbf{x}).
\end{cases}
\end{equation}
where the right hand sides are rational functions, in reduced form, in $\KK(\mathbf{u},\mathbf{x})$ (i.e. no dependencies on the derivatives of $\mathbf{u}$ are allowed). 
\end{definition}

Let us note that in this paper we only treat the case of a single output variable $y$, but in general there may be several outputs.

\begin{remark}
    Our notion of a rational system corresponds to that in \cite[Definition 3.1]{nvemcova2009realization} with the variety $X$ being the whole space $\mathbb{K}^n$ and the input space $U$ being some neighborhood of the origin. 
\end{remark}

We  use the notation $\mathbf{x}'=\mathbf{p}(\mathbf{u},\mathbf{x}), y=q(\mathbf{u},\mathbf{x})$ for~\eqref{eq-realization}. 

\begin{definition} \label{def:realizations}
    A rational system $\Sigma$ as in \eqref{eq-realization} defines a prime differential ideal~\cite[Lemma 3.2]{hong2020global} $$I_\Sigma := \langle x_i' \cdot \denom(p_i)-\num(p_i), y \cdot \denom(q)-\num(q) \rangle : Q^{(\infty)}$$ in the ring of differential polynomials $\KK[\mathbf{u}^{(\infty)},\mathbf{x}^{(\infty)},y^{(\infty)}]$, where $Q$ is the common denominator of all right hand sides, and $\num$ and $\denom$ denote the numerator and denominator, respectively. 
There exists an irreducible differential polynomial $F \in \KK[\mathbf{u}^{(\infty)},y^{(\infty)}]$ such that~\cite[Remark 2.20]{dong2021differential}
\begin{equation}\label{eq-saturation}
\langle F \rangle^{(\infty)} : S_F^{\infty} = I_\Sigma \cap \KK[\mathbf{u}^{(\infty)},y^{(\infty)}]
\end{equation}
where $S_F$ is the \emph{separant} $\frac{\partial\,F}{\partial y^{(n)}}$ of $F$ or order $n$, and $I : a^{\infty}$ denotes the \emph{saturation} $\{ G \in R \mid \exists N \in \NN_0 : a^N\,b \in I \}$ of an ideal $I$ in a ring $R$ w.r.t. $a \in R$. 
We call the implicit equation $F=0$ defined by such a differential polynomial $F$ the \textit{(differential-algebraic) input-output equation} of $\Sigma$; shortly the IO-equation. 
Conversely, for a given irreducible differential polynomial $F$, a system $\Sigma$ of the form~\eqref{eq-realization} such that $F=0$ is its IO-equation is called a (rational) \textit{realization} of $F$. 
Let us note that an IO-equation of a given realization is unique up to multiplication with units, but there might be various realizations for a given irreducible differential polynomial or none at all. If at least one rational realization exists, we say that $F$ is \textit{(rationally) realizable}. We will omit the word ``rationally'' for the rest of the paper.
\end{definition}

\begin{problem}
    By a \emph{rational realization problem} we mean the following problem. Given $F \in \mathbb{K}[\mathbf{u}^{(\infty)}, y^{(\infty)}]$, find a system $\Sigma$ as in \eqref{eq-realization} such that $F=0$ is its IO-equation. From the point of view of differential algebra, this problem is the inverse problem of \emph{differential elimination} of state variables $\mathbf{x}$ from $\Sigma$. In this paper, we study the rational realization problem with additional properties: we ask for a system $\Sigma$ that is globally observable and/or defined over $\mathbb{R}$. Precise statements of these problems are given in Sections \ref{sec-proper} and \ref{sec-real}.
\end{problem}

\begin{example}
We will use \cite[Example 6.3]{pavlov2022realizing} to illustrate the introduced concepts.
Consider the following modified predator-prey model:
\begin{equation*}
\Sigma_0 = 
\begin{cases}
  x_1' = k_1 x_1 - k_2 x_1 x_2,\\
  x_2' = -k_3 x_2 + k_4 x_1^2 x_2 + k_5 u,\\
  y = x_1^2,
\end{cases}
\end{equation*}
where $k_1, \ldots, k_5$ are constants.
Using the software~\cite{dong2021differential}, one computes the following IO-equation for $\Sigma_0$ as in \eqref{eq-saturation}:
$$F := y y'' - 2k_1 k_3 y^2 + 2k_1 k_4 y^3 + k_3 y y' + 2k_2 k_5 y^2 u - k_4 y^2 y' - (y')^2 = 0.$$
The system $\Sigma_0$ is a realization of $F$. As it turns out, $\Sigma_0$ is not a globally observable realization. 
The realization
\begin{equation*} \Sigma =
\begin{cases}
  x_1' = 2k_1 x_1 - 2k_2 x_1 x_2,\\
  x_2' = -k_3 x_2 + k_4 x_1 x_2 + k_5 u,\\
  y = x_1,
\end{cases}
\end{equation*}
obtained in \cite{pavlov2022realizing}, however, is a globally observable realization of $F$.
For real constants $k_1, \ldots, k_5$, $\Sigma$ is also a real realization. For $\mathbf{k}=(1,\mathrm{i},1,1,\mathrm{i})$, however, $\Sigma|_{\mathbf{k}}$ is not a real realization (and neither is $\Sigma_0|_{\mathbf{k}}$) of
$$F|_{\mathbf{k}} = y y'' - 2y^2 + 2y^3 + y y' - 2y^2 u - y^2 y' - (y')^2 = 0.$$
The system
\begin{equation*} \Sigma|_{\mathbf{k}} =
\begin{cases}
  x_1' = 2x_1 - 2\mathrm{i}x_1 x_2,\\
  x_2' = -x_2 + x_1 x_2 + \mathrm{i}u,\\
  y = x_1,
\end{cases}
\end{equation*}
can be transformed by replacing $(x_1, x_2)$ with $(s_1(\mathbf{x}),s_2(\mathbf{x}))=(x_1,\mathrm{i}\,x_2)$ into the real realization (see~\eqref{eq-rerealization})
\begin{equation*}
\begin{cases}
  x_1' = 2x_1+2x_1x_2,\\
  x_2' = -x_2+x_1x_2+u,\\
  y = x_1.
\end{cases}
\end{equation*}
\end{example}

\begin{remark}
    From now on we will focus on the case of a single input variable $u$. This is done in order not to overload the notation of the paper. Our techniques do not rely on there being a single input-variable $u$, and we expect that the same results hold in the presence of multiple inputs.
\end{remark}

\begin{remark}\label{rem-Lp}
    The same letter $n$ in the realization and the IO-equation in Definition \ref{def:realizations} is justified by~\cite[Theorem 3.2]{pavlov2022realizing}, which holds by the same proof for every field of characteristic zero. Alternatively, one can see this by implicitizing the following parametrization~\eqref{eq:parametrization}.

If $F$ has a realization $\mathbf{x}'=\mathbf{p}(u,\mathbf{x}), y=q(u,\mathbf{x})$, then 
\begin{equation}\label{eq:parametrization}
\mathbf{P}=(q,\LD_{\mathbf{p}}(q),\ldots,\LD_{\mathbf{p}}^n(q)),
\end{equation}
where $\LD_{\mathbf{p}}(q)=\sum_{i=1}^n p_i\,\partial_{x_i}q + D_u(q)$ is the Lie-derivative of $q$ w.r.t. $\mathbf{p}$, $\LD_{\mathbf{p}}^i$ is the iterative application of $\LD_{\mathbf{p}}$ $i$-many times, and $D_u$ is defined as the differential operator $D_u(q)=\sum_{j \ge 0} u^{(j+1)} \cdot \partial_{u^{(j)}}q$, defines a parametrization of $\Va(F)$. 
Thus, we have found a necessary condition for the existence of realizations. 
Note that the construction of the parametrization~\eqref{eq:parametrization} from the realization does not require field extensions and if $\mathbf{p},q$ are real, then also $\mathbf{P}$ is real. Similarly, if $\mathbf{p},q$ are polynomial, then also $\mathbf{P}$ is polynomial. Polynomial realizations are an interesting special case of realizations, but will not be studied in this paper.

In~\cite[Theorem 3.1]{forsman1993rational} it is shown that a realization (over $\overline{\KK}$) of $F \in \KK[y^{(\infty)}]$ exists if and only if the corresponding hypersurface $\Va(F)$ is unirational. 
It is not true, however, that every parametrization of $F \in \KK[u^{(\infty)},y^{(\infty)}]$ over $\overline{\KK}(u^{(\infty)})$ leads to a realization (see Example \ref{ex-nonrealizable}).
\end{remark}

\begin{lemma}[Lemma 3.1 in~\cite{pavlov2022realizing}]\label{lem:realizibility}
Let $F \in \KK[u^{(\infty)},y^{(\infty)}]$ be an irreducible polynomial of order $n$ w.r.t. $y$. Then, $F$ is realizable if and only if there exists a rational parametrization $\mathbf{P} =(P_0,\ldots,P_n) \in \overline{\KK}(u^{(\infty)})(\mathbf{x})^{n+1}$ of $\Va(F)$ such that $P_0 \in \overline{\KK}(u)(\mathbf{x})$ and 
\begin{equation}\label{eq-param}
\begin{array}{rcl}
\mathbf{z}&=& \mathcal{J}(P_0,\ldots,P_{n-1})^{-1} \cdot (P_1-D_u(P_0),\ldots,P_n-D_u(P_{n-1}))^T\end{array}
\end{equation}
is in $\overline{\KK}(u,\mathbf{x})^{n}$ where $\mathcal{J}$ denotes the Jacobian (w.r.t. $\mathbf{x}$). In the affirmative case, the realization is $\mathbf{x}'=\mathbf{z},y=P_0$.
\end{lemma}

Let us note that formula~\eqref{eq-param} is similar to that in~\cite[Lemma 3.1]{pavlov2022realizing}, except that we additionally invert the Jacobian. Let us assume that $\mathcal{J}(P_0,\ldots,P_{n-1})$ is singular. Then, by~\cite[Theorem 2.2]{ehrenborg1993apolarity}\footnote{The theorem is stated over $\mathbb{C}$ but the proof works for every field of characteristic zero.}, there exists $G \in \overline{\KK}(u^{(\infty)})[z_0,\ldots,z_{n-1}]$ such that $G(P_0,\ldots,P_{n-1})=0$. 
Since $F$ is assumed to be irreducible, $G(y,\ldots,y^{(n-1)}) \in \langle F \rangle$. 
On the other hand, since $G(y,\ldots,y^{(n-1)})$ is of order at most $n-1$ w.r.t. $y$, $G \notin I_\Sigma$ (cf.~\eqref{eq-saturation}), a contradiction to $\langle F \rangle \subset I_\Sigma$.

\begin{remark}\label{rem:inverse}
The construction of a realization and the corresponding parametrization are connected as follows. If $\mathbf{P}$ is a parametrization of $F$ such that the condition in Lemma~\ref{lem:realizibility} is fulfilled, then the realization given by $\mathbf{z},P_0$ is such that the corresponding parametrization $(P_0,\LD_{\mathbf{z}}(P_0),\ldots,\LD_{\mathbf{z}}^n(P_0))$ is equal to $\mathbf{P}$.
\end{remark}

\begin{remark}\label{rem-param}
Finding rational parametrizations of the hypersuface $\mathbb{V}(F)$ is in general a very difficult problem. However, when $\dim \mathbb{V}(F) = 1$ or $2$, i.e. when $\mathbb{V}(F)$ is a plane curve or a surface in three-space, this problem becomes algorithmic. For an overview of algorithms for finding rational parametrizations of curves and surfaces, see e.g.~\cite{sendra2008rational} and~\cite{schicho1998rational}. These algorithms are implemented in computer algebra systems such as \texttt{Maple}~\cite{deconinck2010computing} and \texttt{MAGMA}~\cite{beck2008software}. The algorithms for finding realizations of IO-equations presented in this paper rely on these parametrization algorithms, and can also be implemented in computer algebra systems for practical use.
\end{remark}

We will now present several general results on rational realizations of IO-equations that will be useful for us later.

\begin{proposition}\label{prop:ordern}
Let $F \in \KK[u^{(\infty)},y^{(\infty)}]$ be irreducible and of order $n$ w.r.t. $y$. If the order of $F$ w.r.t. $u$ is greater than $n$, then $F$ is not realizable.
\end{proposition}
\begin{proof}
Suppose $F$ is realizable but the order of $F$ w.r.t. $u$ is greater than $n$.
Let $\mathbf{x}'=\mathbf{p}(u,\mathbf{x}), y=q(u,\mathbf{x})$ be a rational realization of $F$.
By~\cite[Theorem 3.2]{pavlov2022realizing}, it has at most $n$ states, and the corresponding parametrization is given by \[\mathbf{P}(u,\ldots,u^{(n)},\mathbf{x})=(q,\LD_{\mathbf{p}}(q),\ldots,\LD_{\mathbf{p}}^n(q))\]
which is of order $n$ w.r.t. $u$. 
By implicitizing $\mathbf{P}$, i.e. computing the (algebraic) elimination ideal $I \cap \CC[u,\ldots,u^{(n)},y_0,\ldots,y_n]$ with
\[ I := \langle \denom(P_i) \cdot y_i - \num(P_i) \text{ for } i \in \{0,\ldots,n\}, \denom(P_0) \cdots \denom(P_n) \cdot w - 1 \rangle ,\]
we obtain an equation $G(y,\ldots,y^{(n)},u,\ldots,u^{(n)})=0$ of order at most $n$ w.r.t. $u$ with $\mathbf{x}'=\mathbf{p}(u,\mathbf{x}), y=q(u,\mathbf{x})$ as a realization. 
By~\cite[Remark 4]{dong2021differential}, however, $F=\lambda \cdot G$ for some $\lambda \in \CC$, which means the order of $F$ w.r.t. $u$ is also at most $n$, leading to a contradiction.
\end{proof}

In the following, based on Proposition~\ref{prop:ordern}, we define the order of $F$ as the order of $F$ w.r.t. the output $y$ and assume that the order of $F$ w.r.t. the input $u$ is at most the order w.r.t. $y$. 
Moreover, we will omit dependencies on $u$ and its derivatives in intermediate steps in order to make the paper more readable.

\begin{example}\label{ex-nonrealizable}
Consider $F=(y'-uy)^3+uy^2$ with the parametrization $\mathbf{P}=\left( \frac{u}{(u-x)^3},\frac{ux}{(u-x)^3} \right)$. 
Then~\eqref{eq-param} is
$$z = \frac{ux(u-x)+(2u+x)u'}{3u}$$
and does not lead to a realization because $z$ effectively depends on $u'$.
\end{example}

It is in general hard to verify whether given $F \in \KK[u^{(\infty)},y^{(\infty)}]$ is realizable by only using the condition in Lemma~\ref{lem:realizibility}. For instance, we did not show that $F$ in Example~\ref{ex-nonrealizable} is not realizable and just know that one particular parametrization $\mathbf{P}$ does not correspond to a realization. In~\cite{pavlov2022realizing}, the authors give necessary and sufficient conditions on the parametrizations of $F$ for some special cases. 
Let us recall them here.
\begin{proposition}[Prop. 3.5 in~\cite{pavlov2022realizing}]\label{prop:uorderzero}
Let $F \in \KK[u,y^{(\infty)}]$ be irreducible and of order $n$. Then there exists a rational realization of $F$ if and only if $\Va(F)$ has a rational parametrization $\mathbf{P} \in \overline{\KK}(u)(\mathbf{x})^{n+1}$ such that $P_0,\ldots,P_{n-1} \in \overline{\KK}(\mathbf{x})$.
\end{proposition}
\begin{proposition}[Prop. 3.6 in~\cite{pavlov2022realizing}]\label{prop:uorderone}
Let $F \in \KK[u,u',y^{(\infty)}]$ be irreducible and of order $n$. Then there exists a rational realization of $F$ if and only if $\Va(F)$ has a rational parametrization $\mathbf{P} \in \overline{\KK}(u,u')(\mathbf{x})^{n+1}$ such that $P_{n-1} \in \overline{\KK}(u)(\mathbf{x})$, $P_{n} \in \overline{\KK}(u,u')(\mathbf{x})$ with $\partial_{u'}P_{n}=\partial_{u}P_{n-1}$ and, for $n>1$, $P_0,\ldots,P_{n-2} \in \overline{\KK}(\mathbf{x})$.
\end{proposition}
The differential polynomial $F$ in Example~\ref{ex-nonrealizable} does not fulfill the condition in Proposition~\ref{prop:uorderzero}, because $P_0$ would have to be independent of $u$, i.e., the indeterminate $u$ would have to not appear in $P_0$. We remark that such a parametrization, without radicals in $u$, does not exist and hence, there exists no realization of $F$ in this example.

\begin{definition}
Let $F \in \KK[u^{(\infty)},y^{(\infty)}]$ be irreducible and of order $n$, and let $\mathbf{P}$ be a rational parametrization of $\Va(F)$. We say that $\mathbf{s} \in \overline{\KK}(\mathbf{x})^n$ is a \emph{Lie-suitable reparametrization} of $\mathbf{P}$ if $\mathbf{s}$ defines a reparametrization, i.e. its Jacobian-matrix w.r.t. $\mathbf{x}$ has full rank. 
Note that $\mathbf{s}$ is independent of $u$ and its derivatives.
\end{definition}

The following proposition gives a relation among the realizations of the same IO-equation.

\begin{proposition}\label{prop:reparametrizationHigherOrder}
Let $F \in \KK[u^{(\infty)},y^{(\infty)}]$ be irreducible and of order $n$.
Let $\mathbf{x}'=\mathbf{p}(u,\mathbf{x}), y=q(u,\mathbf{x})$ and $\mathbf{x}'=\mathbf{f}(u,\mathbf{x}), y=g(u,\mathbf{x})$ be realizations of $F$ such that the corresponding parametrizations $\mathbf{P}=(q,\LD_{\mathbf{p}}(q),\ldots,\LD^n_{\mathbf{p}}(q))$ and $\mathbf{Q}=(g,\LD_{\mathbf{f}}(g),\ldots,\LD_{\mathbf{f}}^n(g))$ fulfill $\mathbf{P}(\mathbf{s})=\mathbf{Q}(\mathbf{x})$ for some $\mathbf{s} \in \overline{\KK}(u,\ldots,u^{(n)},\mathbf{x})^n$. 
Then $\mathbf{s} \in \overline{\KK}(x)^n$ is Lie-suitable. 
Moreover, for every Lie-suitable reparametrization $\mathbf{s}$,
\begin{equation}\label{eq-rerealization}
\mathbf{x}'= \mathcal{J}(\mathbf{s}(\mathbf{x}))^{-1} \cdot \mathbf{p}(u,\mathbf{s}), \, y=q(u,\mathbf{s})
\end{equation}
is another realization of $F$.
\end{proposition}
\begin{proof}
Let us consider the parametrization $\mathbf{P}(\mathbf{s})=(q(u,\mathbf{s}),\LD_{\mathbf{p}}(q),\ldots,\LD^n_{\mathbf{p}}(q))$. Since $p(u,\mathbf{s})=g(u,\mathbf{x})$, one has that $\mathbf{s}$ is independent of $u',\ldots,u^{(n)}$.
Let us use the notation $\partial_u \mathbf{s} = (\partial_us_1,\ldots,\partial_us_n)^T$. 
Then~\eqref{eq-param} is
\begin{align*}
& (\mathcal{J}(P_0,\ldots,P_{n-1})(\mathbf{s}) \cdot \mathcal{J}(\mathbf{s}))^{-1} \cdot (P_1(\mathbf{s}) - D_u(P_0(\mathbf{s})),\ldots,P_n(\mathbf{s})-D_u(P_{n-1}(\mathbf{s})))^T \\ 
&= \mathcal{J}(\mathbf{s})^{-1} \cdot \mathcal{J}(P_0,\ldots,P_{n-1})(\mathbf{s})^{-1} \cdot (\LD_{\mathbf{p}}(P_0)(\mathbf{s})-D_u(P_0(\mathbf{s})),\ldots,\LD_{\mathbf{p}}(P_{n-1})(\mathbf{s})-D_u(P_{n-1}(\mathbf{s})))^T \\
&= \mathcal{J}(\mathbf{s})^{-1} \cdot \mathcal{J}(P_0,\ldots,P_{n-1})(\mathbf{s})^{-1} \cdot (\mathcal{J}(P_0,\ldots,P_{n-1})(\mathbf{s}) \cdot (\mathbf{p}- u' \cdot \partial_u \mathbf{s})) \\
&= \mathcal{J}(\mathbf{s})^{-1} \cdot (\mathbf{p} - u' \cdot \partial_u \mathbf{s}),
\end{align*}
where in the second step we have used the chain rule applied to $D_u$ and the independence of $\mathbf{s}$ from derivatives of $u$. The result has to be independent of $u'$. Thus, $\partial_u \mathbf{s}=0$. For such $\mathbf{s} \in \overline{\KK}(\mathbf{x})^n$, we obtain a new realization of $F$ and it is of the form~\eqref{eq-rerealization}.
\end{proof}

\begin{definition}
    Based on Proposition~\ref{prop:reparametrizationHigherOrder}, we may call a realization as in~\eqref{eq-rerealization} a \textit{reparametrization} of the given realization $\mathbf{x}'=\mathbf{p}(u,\mathbf{x}), y=q(u,\mathbf{x})$.
\end{definition}

\begin{remark}\label{rem:nonrationalreparametrization}
In Proposition~\ref{prop:reparametrizationHigherOrder}, all objects are assumed to be rational in their arguments. 
By essentially the same proof, it can be generalized as follows. 

Let $\mathbf{x}'=\mathbf{p}(u,\mathbf{x}), y=q(u,\mathbf{x})$ be a realization and let $\mathbf{s}(\mathbf{x}) \in \overline{\KK(\mathbf{x})}^n$ be such that $\mathbf{Q}(\mathbf{x}):=\mathbf{P}(\mathbf{s})$ is rational.
Then $\mathbf{z}$ in~\eqref{eq-param}, computed for $\mathbf{Q}$, is rational as well and~\eqref{eq-rerealization} gives a realization of $F$.
\end{remark}

As explained in~\cite{pavlov2022realizing}, when $n=1$ or the given IO-equation $F$ is independent of derivatives of $u$, the decision of whether $F$ is realizable is algorithmic. In this paper, we mainly study observable and real realizations of these types of IO-equations.

\section{Observable realizations}\label{sec-proper}
In this section, we investigate realizations such that the corresponding parametrization is proper. These realizations have the special properties that all other realizations can be found from them by means of reparametrizations; and the states $\mathbf{x}$ are ``observable'', an important property in control theory. The cases where the IO-equation is independent of $u$ or of first-order are special and can be treated algorithmically.

\begin{remark}\label{rem-observability}
For a realization $\mathbf{x}'=\mathbf{p}(u,\mathbf{x}), y=q(u,\mathbf{x})$ such that the corresponding parametrizations is proper, it holds that $$\overline{\KK}(u,\ldots,u^{(n)})(q,\LD_q(\mathbf{p}),\ldots,\LD_q^n(\mathbf{p})) = \overline{\KK}(u,\ldots,u^{(n)})(\mathbf{x}).$$
In control theory, a common question is whether the states $\mathbf{x}$ are \textit{(globally) observable}, that is, whether $\overline{\KK}(u^{(\infty)})(q,\LD_q(\mathbf{p}),\ldots) = \overline{\KK}(u^{(\infty)})(\mathbf{x})$~\cite[Proposition 3.4]{hong2020global}\footnote{Let us remark that  observability of a state $x_i(t)$, defined as in e.g.~\cite{birk1993computer}, is equivalent to identifiability of the initial value $x_i(0)$ treated as a parameter, which is in turn the subject of \cite[Proposition 3.4]{hong2020global}.}. 
By~\cite[Theorem 3.16]{hong2020global}, the properness of the corresponding parametrization is a necessary and sufficient condition for the observability of all states $\mathbf{x}$.
\end{remark}

\begin{remark}\label{rem-algobs}
    The algebraic characterization of global observability relies on \cite{hong2020global}. The set of inputs considered in that paper is a generic (Zariski-open) subset of the set of functions that are analytic in some neighborhood of the origin (see e.g. \cite[Notation 2.4 and Definition 2.5]{hong2020global}). This is the set of inputs for which (the algebraic) Definition \ref{def:observability} of observability makes sense.
\end{remark}

Remarks~\ref{rem-observability} and~\ref{rem-algobs} justify the following definition.

\begin{definition} \label{def:observability}
Let $F \in \KK[u^{(\infty)},y^{(\infty)}]$ be an irreducible differential polynomial of order $n$. Assume that $F$ is realizable. A realization $\mathbf{x}'=\mathbf{p}(u,\mathbf{x}), y=q(u,\mathbf{x})$ is called \textit{(globally) observable} for the set of inputs described in Remark~\ref{rem-algobs} if the corresponding parametrization $(q,\LD_{\mathbf{p}}(q),\ldots,\LD_{\mathbf{p}}^n(q))$ is a proper rational parametrization of $\Va(F)$.
Moreover, if there exists an observable realization of $F$, then we say that $F$ is \textit{observably realizable}.
\end{definition}

Let us note that properness of a given parametrization can always be checked by, for instance, using elimination techniques. For curves and surfaces there are degree-conditions that are easy to verify, see Remark~\ref{rem-necDegreeConditionsProper},~\cite[Theorem 5]{perez2005partial} and~\cite{perez2008univariate}.

\begin{problem}
    The problem treated in this section is that of \emph{observable rational realizability}. It is formulated as follows: Given $F \in \mathbb{K}[u^{(\infty)}, y^{(\infty)}]$ find an \emph{observable} (in the sense of Definition \ref{def:observability}) system $\Sigma$ as in \eqref{eq-realization} such that  $F=0$ is the IO-equation of $\Sigma$.
\end{problem}

\begin{lemma}\label{lem-fieldextesionproper}
Let $V$ be a rational variety over an algebraically closed field $\LL$ of characteristic zero. 
Let $\mathbf{P}(\mathbf{x})$ be a proper parametrization of $V$ and $\mathbf{Q}(\mathbf{x})$ be another parametrization of $V$, not necessarily proper. Let $\KK$ be the smallest subfield of $\LL$ containing the coefficients of $\mathbf{P}$, the coefficients of $\mathbf{Q}$ and the coefficients of a finite set of generators of $V$. Then, there exists a rational reparametrization $\mathbf{s}(\mathbf{x})$ with coefficients in $\KK$ such that $\mathbf{P}(\mathbf{s})=\mathbf{Q}(\mathbf{x})$.
\end{lemma}
\begin{proof}
Since $\mathbf{P}$ is proper, there exists $\mathbf{P}^{-1} \colon V \rightarrow \LL^{\dim(V)}$. Since $V$ is $\KK$-definable (i.e. is given by equations with coefficients in $\mathbb{K}$), $\mathbf{P}$ has coefficients in $\KK$, and using that elimination theory does not extend the ground field, we get that $\mathbf{P}^{-1}$ has coefficients in $\KK$. So, $\mathbf{s}=\mathbf{P}^{-1}(\mathbf{Q}(\mathbf{x}))$ has coefficients in $\KK$ and clearly satisfies $\mathbf{P}(\mathbf{s})=\mathbf{Q}(\mathbf{x})$.
\end{proof}

\begin{theorem}\label{thm:singletonrealizationbasis}
Let $F \in \KK[u^{(\infty)},y^{(\infty)}]$ be irreducible and of order $n$. 
Let $\Sigma=\{\mathbf{x}'=\mathbf{p}(u,\mathbf{x}), y=q(u,\mathbf{x})\}$ be an observable realization of $F$. 
Then every realization of $F$ is found by a reparametrization of $\Sigma$.
\end{theorem}
\begin{proof}
$F$ has coefficients in $\overline{\KK}(u)$ and $\mathbf{P}, \mathbf{Q}$ have coefficients in $\overline{\KK}(u,\ldots,u^{(n)})$; see e.g. comments after formula~\eqref{eq:parametrization}. Then, by applying Lemma~\ref{lem-fieldextesionproper}, there exists $\mathbf{s}\in\overline{\KK}(u,\ldots,u^{(n)})(\mathbf{x})$ such that $\mathbf{P}(\mathbf{s})=\mathbf{Q}(\mathbf{x})$. Now the claim follows from Proposition~\ref{prop:reparametrizationHigherOrder}.
\end{proof}

Theorem~\ref{thm:singletonrealizationbasis} motivates the study of observable realizations, because they generate all other realizations via reparametrizations, similarly to the case of proper parametrizations of algebraic varieties (see Lemma~\ref{lem-fieldextesionproper}). The following proposition deals with the properness of the output provided by Algorithm 1 in~\cite{pavlov2022realizing}.

\begin{proposition}\label{prop-properrealization3}
Let $F \in \KK[u,y^{(\infty)}]$ be irreducible and of order $n$. 
Assume that the parametrization computed in step (S8)a of~\cite[Algorithm 1]{pavlov2022realizing} is proper. Then the output parametrization, generated by ~\cite[Algorithm 1]{pavlov2022realizing}, is also proper and provides an observable realization of $F$.
\end{proposition}
\begin{proof}
We follow the proof of~\cite[Lemma 5.2]{pavlov2022realizing}. 
Assume that the produced parametrization $\mathbf{P}=(P_0,\ldots,P_n)$ is improper. This means that $\overline{\KK}(u)(\mathbf{P}) \subsetneq \overline{\KK}(u)(\mathbf{x})$. 
Therefore, there exists an automorphism $\sigma$ of $\overline{\KK(u,\mathbf{x})}/\KK$ such that $\sigma |_{\overline{\KK}(u,\mathbf{P})} = \text{id}$, and, $\sigma(x_j) \ne x_j$ for some $j \in \{1,\ldots,n\}$. 
For $i \in \{0,\ldots,n-1\}$, 
\begin{equation}\label{eq-sigma}
P_i(\sigma(x_1),\ldots,\sigma(x_n))=P_i(x_1,\ldots,x_n) \in \overline{\KK}(\mathbf{x}).
\end{equation}
Assume that $P_0,\ldots,P_{n-1}$ are algebraically dependent, i.e. there exists an irreducible $G \in \overline{\KK}[y_0,\ldots,y_{n-1}] \setminus \{0\}$ such that $G(P_0,\ldots,P_{n-1})=0$.
Since the rank of $\mathcal{J}(P_0,\ldots,P_{n-1})$ is $n-1$ (see comment after Lemma~\ref{lem:realizibility}), the image $(P_0,\ldots,P_{n-1})(\overline{\KK}^{n})$ is Zariski dense in $\overline{\KK}^{n}$. So, $G$ vanishes on a dense subset of $\overline{\KK}^{n}$ and therefore is constantly zero, a contradiction. 
Thus, $P_0,\ldots,P_{n-1}$ are algebraically independent. 
From~\eqref{eq-sigma} we see that $\sigma(x_1),\ldots,\sigma(x_n)$ are independent of $u$ and elements in $\overline{\KK(\mathbf{x})}$. 
Since $\sigma$ fixes $P_n$ and $u$, and $u$ is transcendental over $\overline{\KK(\mathbf{x})}$, $\sigma$ fixes the coefficients of $P_n$. 
Let $\mathbf{Q}=(P_0,\ldots,P_{n-1},Q_n) \in \KK(\mathbf{x})^{n+1}$ be the proper parametrization computed in step (S8)a. 
Its last component $Q_n$ is a $\QQ$-linear combination of the coefficients of $P_n$ and thus, $\sigma(Q_n)=Q_n$. 
This, however, contradicts to the properness of $\mathbf{Q}$.
\end{proof}

Proposition~\ref{prop-properrealization3} ensures the existence of an observable realization under the assumption that step (S8)a of~\cite[Algorithm 1]{pavlov2022realizing} provides a proper parametrization. If the hypersurface appearing in that step is of dimension larger than two, then it could happen that such a proper parametrization does not exist. 
Nevertheless, the cases $n=1$ and $n=2$ are special, because a proper parametrization of a unirational curve or surface, respectively, can always be found (see~\cite[Theorem 4.10]{sendra2008rational} and Castelnuovo's Theorem~\cite{castelnuovo1939,schicho1998rational}, respectively) leading to the following result.

\begin{theorem}\label{thm:properrealization2}
Let $F \in \KK[u,y,y']$ or $F \in \KK[u,y,y',y'']$ be irreducible and realizable. 
Then $F$ is observably realizable.
\end{theorem}
\begin{proof}
By Proposition~\ref{prop-properrealization3} and the fact that a surface over $\overline{\KK}$ is unirational if and only if there exists a proper parametrization, there exists a rational parametrization $\mathbf{P} \in \overline{\KK}(x_1,x_2)^2 \times \overline{\KK}(u)(x_1,x_2)$ if and only if there exists a proper one. 
Then the statement follows from Remark~\ref{rem:inverse}.
\end{proof}

For some realizable IO-equations, there might not exist an observable rational realization. This resembles the negative answer to the L\"uroth problem in classical algebraic geometry, see e.g.~\cite{artin1972some}.
Consider a unirational but not rational hypersurface defined by $F \in \QQ[y,y',y'',y^{(3)}]$, such as~\cite{segre1960variazione} $$F=y^4+y+(y')^4-6(y'')^2(y^{(3)})^2+(y'')^4+(y^{(3)})^4+(y^{(3)})^3.$$ 
Then the unirational parametrization $\mathbf{P} \in \overline{\QQ}(x_1,x_2,x_3)$ gives a realization as in~\eqref{eq-realization}; note that all components are independent from $u$, but there cannot be an observable realization.

So, we have given a complete answer to the cases where the IO-equation is of order at most two w.r.t. $y$ and order zero w.r.t. $u$. The question of whether every realizable IO-equation of order one, w.r.t. both $y$ and $u$, is observably realizable is answered positively in the next section. 
The case of second order IO-equations remains as an open problem.

\subsection{First order IO-equations}\label{sec-proper2}
We now study the case of first-order IO-equations where $F \in \KK[u,u',y,y']$ effectively depends on $u'$. 
If $F$ is independent of $u'$, we have shown in Theorem~\ref{thm:properrealization2} that if $F$ is realizable, then it is observably realizable. 
We generalize this result here.

Let us note that in the proof of~\cite[Proposition 5.5]{pavlov2022realizing} it is shown that every realization can be obtained by Lie-suitable reparametrizations $s \in \overline{\KK}(x)$ from the realizations produced by~\cite[Algorithm 2]{pavlov2022realizing}. 
As demonstrated in Example~\ref{ex-severalfactors}, not every possible output of this algorithm leads to an observable realization.

\begin{remark}\label{rem-necDegreeConditionsProper}
For deciding whether a given realization $x'=p(u,x), y=q(u,x)$ is observable, one can compute the corresponding parametrization $\mathbf{P}=(q,\LD_p(q))$ and then check whether the so-called tracing index, i.e. the cardinality of a generic fiber of the map induced by $\mathbf{P}$, is one~\cite[Theorem 4.30]{sendra2008rational}; or whether the degree conditions~\cite[Theorem 4.21]{sendra2008rational}
$$\deg_x(q) = \deg_{y'}(F), \, \deg_x(\LD_p(q)) = \deg_y(F)$$
are fulfilled for non-zero $q$. 
Note that $\deg_x(\LD_p(q)) \le \deg_x(q)+\deg_x(p)+1$.
\end{remark}

If the parametrization $\mathbf{P}$ derived from a given realization is improper, one may try to transform it into a proper one and thus get an observable realization. 
This is always possible as shown in what follows. We start with a technical lemma.

\begin{lemma}\label{lem-propernessprojection}
Let $\LL$ be an algebraically closed field, and let $\tilde{Q}:=(\tilde{Q}_1,\tilde{Q}_2,\tilde{Q}_3) \in \LL(x)^3$ be a proper parametrization of an algebraic space curve $\Cu$ over $\LL$ with $\tilde{Q}_1 \notin \LL$. Let $c$ be a transcendental element over $\LL$. We consider $\mathbf{Q}:=(\tilde{Q}_1, \tilde{Q}_2 + c\,\tilde{Q}_3)$ which is a parametrization of a plane curve $\Cu_1$ over $\overline{\LL(c)}$. Then $\mathbf{Q}$ is a proper parametrization.
\end{lemma}
\begin{proof}
Define $G_i(a,b):=\num(Q_i(a)-Q_i(b))$ for $i \in \{1,2\}$. 
Note that since $\tilde{Q}_1$ is non-constant, $G_1$ is non-zero. 
Consider $G:=\gcd(G_1,G_2) \in \LL(c)[a,b]$. 
Let $\tilde{G}_i:=\num(\tilde{Q}_i(a)-\tilde{Q}_i(b))$ for $i \in \{1,2,3\}$ and $\tilde{G}:=\gcd(\tilde{G}_1,\tilde{G}_2,\tilde{G}_3)$. 
Since $G$ divides $\tilde{G}_1$, we can assume that $G$ is independent of $c$. 
Because $G \in \LL[a,b]$ divides $G_2$, it divides $\tilde{G}_2$ and $\tilde{G}_3$. 
Thus, $G$ divides $\tilde{G}$. 
By~\cite[Section 2]{perez2013partial}, the degree of $\tilde{G}$ is one. 
Thus, the degree of $G$ is one as well and $\mathbf{Q}$ is proper.
\end{proof}

\begin{theorem}\label{thm:firstorderproper}
Let $F \in \KK[u,u',y,y']$ be irreducible and realizable. 
Then $F$ is observably realizable.
\end{theorem}
\begin{proof}
Let $x'=p(u,x), y=q(u,x)$ be a realization of $F$ with corresponding parametrization $\mathbf{P}=(q,\LD_p(q))$. 
If $\mathbf{P}$ is proper, the claim follows. 
So let us assume that $\mathbf{P}$ is improper. 
Let us write $F=\sum_{i=0}^d f_i(u,y,y')\,u'^i$ and define $g_i(u,z_0,z_1,z_2)$ as the $i$th coefficient of $\sum_{i=0}^d f_i(u,z_0,z_1+u'\,z_2)\,u'^i$ seen as polynomial in $u'$. 
Since $F(\mathbf{P})=0$, we obtain that $\tilde{P}:=(q,\partial_xq \cdot p, \partial_uq)$ fulfills $g_i(\tilde{P})=0$ for every $i \in \{0,\ldots,\ell\}$ and defines an irreducible curve on the variety $\Va(g_0,\ldots,g_\ell)$, denoted by $\Cu_g$. Note that since $q$ is non-constant, because $\mathbf{P}$ is a parametrization, then $\tilde{P}$ is also a parametrization and in particular non-constant. 
By a version of~\cite[Lemma 4.17]{sendra2008rational} for space-curves, there exists a proper parametrization $\tilde{Q} \in \overline{\KK}(u)(x)^3$ of $\Cu_g$ and a reparametrization $s \in \overline{\KK}(u)(x)$ such that $\tilde{Q}(s)=\tilde{P}(x)$. 
Note that no field extension is necessary for obtaining $s$ and $\tilde{Q}$~\cite[Remark 1]{perez2013partial}.
Let us consider $$\frac{d}{du}\tilde{Q}_1(u,s) = \partial_u \tilde{Q}_1(u,s) + \partial_x \tilde{Q}_1(u,s) \cdot \partial_u s = \partial_u q(u,x).$$
Since $\tilde{Q}_3(u,s)=\partial_u q(u,x)$, it holds that 
\begin{equation}\label{eq-SRGS}
\partial_u s = \frac{\tilde{Q}_3(u,s)-\partial_u \tilde{Q}_1(u,s)}{\partial_x \tilde{Q}_1(u,s)} \in \overline{\KK}(u,s).
\end{equation}
Thus, the reparametrization $s$ is a so-called strong rational general solution (see~\cite{vo2018deciding}) of this first-order differential equation in $u$ (with transcendental constant $x$). 
By~\cite[Theorem 5.2]{vo2018deciding},~\eqref{eq-SRGS} is either a Riccati equation or linear. There exists a strong rational general solution $r$ of~\eqref{eq-SRGS} with $\deg_x(r)=1$ (see e.g.~\cite[Section A1.2, A1.3]{murphy2011ordinary}) and thus, $r$ is a M\"obius transformation (seen as an element in $\overline{\KK}(u)(x)$). 
Define $\mathbf{Q}:=(\tilde{Q}_1(r),\tilde{Q}_2(r)+u'\,\tilde{Q}_3(r)) \in \overline{\KK}(u,x) \times \overline{\KK}(u,u',x)$. 
By construction, $\mathbf{Q}=(Q_1,Q_2)$ is a parametrization of $\Cu(F)$. 
Moreover, $\partial_{u'}Q_2 = \tilde{Q}_3(r) = \partial_uQ_1$ follows from~\eqref{eq-SRGS} and, by Proposition~\ref{prop:uorderone}, $\mathbf{Q}$ corresponds to a realization which is, by Lemma~\ref{lem:realizibility}, given as
\begin{equation}\label{eq-properrealization}
x'=\frac{\tilde{Q}_2(r)+u'\,\tilde{Q}_3(r)-u'\,(\partial_u\tilde{Q}_1(r)+\partial_x\tilde{Q}_1(r) \cdot \partial_ur)}{\partial_x\tilde{Q}_1(r) \cdot \partial_xr} = \frac{\tilde{Q}_2(r)}{\partial_x\tilde{Q}_1(r) \cdot \partial_xr},\,
y=\tilde{Q}_1(r).
\end{equation}
Since $\tilde{Q}(r)$ is a proper parametrization, by Lemma~\ref{lem-propernessprojection}, also $\mathbf{Q}$ is proper.
\end{proof}

Let us note that, in the notation of Theorem~\ref{thm:firstorderproper}, as a consequence of Theorem~\ref{thm:singletonrealizationbasis}, the improper parametrization $\mathbf{P}=(q,\LD_p(q))$ is obtained as a Lie-suitable reparametrization of $\mathbf{Q}$. 
Finding the proper parametrization $\mathbf{Q}$, corresponding to a realization, can either be done by following the proof of Theorem~\ref{thm:firstorderproper} or by an adapted version of~\cite[Theorem 6.4]{sendra2008rational}. 
We choose the latter approach because otherwise we still have to use~\cite[Theorem 6.4]{sendra2008rational} for computing the proper parametrization of the space curve as one of the intermediate steps. 
Define 
\[G_i^P(w,x)=\num(P_i)(w)\,\denom(P_i)(x)-\num(P_i)(x)\,\denom(P_i)(w) \in \overline{\KK}(u,u')[w,x]\] and set \[G^P(w,x)=\gcd(G_1^P,G_2^P) \in \overline{\KK}(u,u')[w,x].\]

\begin{theorem}\label{thm-properrealization}
Let $F \in \KK[u,u',y,y']$ be irreducible with a realization $x'=p(u,x), y=q(u,x)$ and corresponding parametrization $\mathbf{P}=(q,\LD_p(q))$. 
Let
\begin{equation}\label{eq-R}
R:=\left\{\frac{a\,G^P(\alpha,x)+b\,G^P(\beta,x)}{c\,G^P(\alpha,x)+d\,G^P(\beta,x)} \mid \alpha,\beta,a,b,c,d\in\overline{\KK(u)}, G^P(\alpha,\beta) \ne 0, ad-bc \ne 0 \right\} .
\end{equation}
Then $R \cap \overline{\KK}(x) \ne \emptyset$ and there exists an observable realization of $F$ with corresponding parametrization $\mathbf{Q}=(g,\LD_f(g))$ such that $\mathbf{Q}(s)=\mathbf{P}(x)$.
\end{theorem}
\begin{proof}
Let us write $\LL_1=\overline{\KK}(u)$, $\LL_2=\overline{\KK}(u,u')$. 
Let $\mathbf{P}=(q,\LD_p(q)) \in \LL_1(x) \times \LL_2(x)$ be expressed in reduced form, i.e. with coprime numerator and denominator in both components. 
Note that, for $i \in \{1,2\}$, $G_i^P\in \LL_i[w,x]$. Thus, $G^P \in \LL_1[w,x]$ and it is sufficient to work over $\LL_1$. 
By~\cite[Theorem 6.4]{sendra2008rational}, for every $r \in R \subset \overline{\LL_1}(x)$ there exists a proper parametrization $\mathbf{Q} \in \overline{\LL_2}(x)^2$, so that $\mathbf{Q}(r)=\mathbf{P}(x)$. 
From Theorem~\ref{thm:firstorderproper} we know that $R \cap \overline{\KK}(x) \ne \emptyset$. 
Choose $r \in R \cap \overline{\KK}(x)$. 
The defining polynomials $g_i(z_1,z_2)$ of $(P_i,r) \in \LL_i(x)^2$ have coefficients in $\LL_i$. 
Since $g_i$ is linear in $z_2$, and $Q_i$ is a root of $g_i$ in $z_2$ (see~\cite[Algorithm Proper-Reparametrization]{sendra2008rational}), it holds that $\mathbf{Q} \in \LL_1(x) \times \LL_2(x)$. 
By Remark~\ref{rem:nonrationalreparametrization}, it holds that $\mathbf{Q}$ defines a realization of $F$. 
The realization is indeed rational because the corresponding parametrization $\mathbf{Q}$ is rational.
\end{proof}

For computing $R \cap \overline{\KK}(x)$ in Theorem~\ref{thm-properrealization}, one can use the following method.\\
Let $r \in R \cap \overline{\KK}(x)$ be expressed as $r=M(x)/N(x)$ for $M,N\in \overline{\KK}[x]$ and $\gcd(M,N)=1$. Since $r \in R$, 
\begin{equation}\label{eq-equality}
M(x) \, (c\,G^P(\alpha,x)+d\,G^P(\beta,x))= N(x) \,(a\,G^P(\alpha,x)+b\,G^P(\beta,x)).
\end{equation}
All rational functions in $R$ have the same degree w.r.t. $x$ (see e.g. Theorem 6.3. in~\cite{sendra2008rational}). That is
\[\ell:=\max\{\deg_x(a\,G^P(\alpha,x)+b\,G^P(\beta,x)),\deg_x(c\,G^P(\alpha,x)+d\,G^P(\beta,x))\}=\max\{\deg_x(M),\deg_x(N)\}.\]
Let us say w.l.o.g. that $\deg_x(M)=\max\{\deg_x(M),\deg_x(N)\}$. Now, since $\gcd(M,N)=1$, by~\eqref{eq-equality}, one gets that $M$ divides $(a\,G^P(\alpha,x)+b\,G^P(\beta,x))$. Then, since $\deg_x(a\,G^P(\alpha,x)+b\,G^P(\beta,x))\leq \deg_x(M)$, we have that 
\[ (a\,G^P(\alpha,x)+b\,G^P(\beta,x))=\lambda M(x) \,\,\text{with } \lambda\in \overline{\KK(u)}. \]
After substituting this in the right hand side of~\eqref{eq-equality} and dividing by $M$, we obtain $(c\,G^P(\alpha,x)+d\,G^P(\beta,x))=\lambda N(x)$. 
In this situation, let $\mathbf{m}:=(m_0,\ldots,m_\ell),\mathbf{n}:=(n_0,\ldots,n_\ell)$ be tuples of new variables. We consider the polynomials 
\[ \begin{array}{l} E_1:= A\,G^P(A_1,x)+B\,G^P(B_1,x)-\lambda \sum_{i=0}^{\ell} m_i x^i \in \overline{\KK}[u,\mathbf{m},\lambda,A,B,A_1,B_1,x], \\
E_2:= C\,G^P(A_1,x)+D\,G^P(B_1,x)-\lambda \sum_{i=0}^{\ell} n_i x^i \in  \overline{\KK}[u,\mathbf{n},\lambda, C,D,A_1,B_1,x], \\
E_3:=Z_1\cdot G^P(A_1,B_1) (AB-CD)-1 \in \overline{\KK}[Z_1,u,A,B,C,D,A_1,A_2], \\
E_4:=Z_2 n_\ell+Z_3 m_\ell -1 \in \overline{\KK}[Z_2,Z_3,m_{\ell},n_{\ell}],
\end{array}\]
where $\lambda, Z_k,A,B,C,D,A_1,B_1$ are new variables. Now, let $\mathcal{V}^*$ be the set containing all non-zero coefficients of $E_1,E_2$ w.r.t. $x$ and let $\mathcal{V}_u:=\mathcal{V}^*\cup \{E_3,E_4\}$.
Now eliminate $u$ in $\mathcal{V}_u$ to obtain an ideal $\mathcal{V}$ with $\Va(\mathcal{V}) \subset \overline{\KK}^{2\ell+12}$.
By construction, $R \cap \overline{\KK}(x) \neq \emptyset$ if and only if $\Va(\mathcal{V}) \neq \emptyset$. 
Moreover, a zero of $\mathcal{V}$ defines an element $r \in R \cap \overline{\KK}(x)$.

Finally, we want to explicitly compute the proper realization. 
Given the improper parametrization $\mathbf{P}(x)$ and $r \in R \cap \overline{\KK}(x)$ as above, this can be done by
\begin{enumerate}
    \item making an ansatz for $\mathbf{Q}(x)$ of degree $\deg_x(\mathbf{Q})=\deg_x(\mathbf{P})/\deg_x(r)$ and $\deg_u(\mathbf{Q})=\deg_u(\mathbf{P})$ with undetermined coefficients, and solving the resulting linear system; or,
    \item by computing the implicit equations of $(P_i,r)$, $i \in \{1,2\}$, which are of the form $g_i(w,x) = \denom(Q_i)(x) - w\,\num(Q_i)(x)$ (see~\cite[Algorithm Proper-Reparametrization]{sendra2008rational}).
\end{enumerate}

\begin{algorithm}[H]
\caption{ObservableRealization}
\label{alg-properrealization}
\begin{algorithmic}[1]
    \REQUIRE An irreducible polynomial $F \in \KK[u,u',y,y']$ over a computable field $\KK$.
    \ENSURE An observable realization of $F$ if it exists.
    \STATE Check whether $F$ is realizable (e.g. by~\cite[Algorithm 2]{pavlov2022realizing}).
    \STATE In the affirmative case, let $x'=p(x,u), y=q(x,u)$ be any realization of $F$.
    \STATE Compute $R$ corresponding to the parametrization $\mathbf{P}=(q,\LD_p(q))$ as in~\eqref{eq-R}.
    \STATE Compute the intersection $R \cap \overline{\KK}(x)$ as described above and choose $r \in R \cap \overline{\KK}(x)$. 
    \STATE Compute $\mathbf{Q}(x)$ with $\mathbf{Q}(r)=\mathbf{P}(x)$ as in~\cite[Algorithm Proper-Reparametrization]{sendra2008rational}.
    \STATE Output the observable realization $x'=f(u,x), y=g(u,x)$ corresponding to $\mathbf{Q}$.
\end{algorithmic}
\end{algorithm}

\begin{theorem}\label{thm-Alg1correct}
Algorithm~\ref{alg-properrealization} is correct.
\end{theorem}
\begin{proof}
Correctness of the algorithm follows from Theorem~\ref{thm-properrealization} together with the correctness of~\cite[Algorithm 2]{pavlov2022realizing}. 
By Remark~\ref{rem:nonrationalreparametrization}, the output is indeed a realization; note that by construction in Theorem~\ref{thm-properrealization}, the right hand sides are indeed rational. 
The termination follows by the termination of each step.
\end{proof}


\begin{corollary}\label{rem:oneoutputproper}
Let $F \in \KK[u,u,y,y']$ be irreducible and realizable. 
Among the finitely many possible outputs of~\cite[Algorithm 2]{pavlov2022realizing} there exists an observable realization.
\end{corollary}
\begin{proof}
As mentioned in the beginning of the current section, every realization of $F$ can be found by a Lie-suitable reparametrization from the outputs of~\cite[Algorithm 2]{pavlov2022realizing} when every pair of factors occurring in steps (S2) and (S2)b is checked. 
Assume that in the outputs there exists no observable realization. 
Note that reparamerization of a non-observable realization again leads to a non-observable realization. 
By Theorem~\ref{thm:firstorderproper}, however, there exists an observable realization.
\end{proof}

\begin{example}\label{ex-severalfactors}
Let us consider the differential polynomial
\begin{align*}
F = \, & 27u^6y^3 - 27u^5y^2y' + 27u^4u'y^3 + 9u^4yy'^2 - 18u^3u'y^2y' + 9u^2u'^2y^3 - 4u^4y^2 - u^3y'^3 \\ & + 3u^2u'yy'^2 - 3uu'^2y^2y' + u'^3y^3 + 4u^3yy' - 4u^2u'y^2 - u^2y'^2 + 4uu'yy' - u'y'^2.
\end{align*}
By applying~\cite[Algorithm 2]{pavlov2022realizing} and using $y_0 = c^3u + c^2$, we obtain the two factors
\begin{align*}
N_1 = 9c^6u^5 + 21c^5u^4 - 6bc^3u^3 + 13c^4u^3 - 7bc^2u^2 - c^3u^2 + b^2u - 2c^2u + b
\end{align*}
and $N_2 = -3c^3u^2 - 2c^2u + b$. For the first factor $N_1$, any rational parametrization will lead to a non-observable realization. 
For $N_2$ and for example $(b(x),c(x)) = (3u^2x^3 + 2ux^2, x)$, however, we find the observable realization $x' = ux, \, y= ux^3 + x^2$.
\end{example}

\begin{example}\label{ex-nonProperRealization}
Let us consider the non-observable realization
$$x'=\frac{1-x}{2u}, \, y=\frac{(1-x)^4}{u^2+(1-x)^6}$$
of some irreducible differential polynomial $F \in \RR[u,u',y,y']$. 
Note that $F$ does not have to be computed, but can be found by implicitizing the parametrization $\mathbf{P}$ corresponding to the given realization. 
$R$ corresponding to $\mathbf{P}$ is given by $G^P(w,x)=(w - 2 + x)(w - x)$. 
We can choose $r:=-x^2+2x \in R \cap \CC(x)$. 
The implicit equations is $(P_1,r)$
$$g_1(z_1,z_2)=z_2^3 - u^2 - 3z_2^2 + 3z_2 - 1 - (-z_2^2 + 2z_2 - 1)\,z_1 = \denom(Q_1)- \num(Q_1) \cdot z_1.$$
Similarly $g_2$ can be found such that $\mathbf{Q}=\left(\tfrac{-z_2^2 + 2z_2 - 1}{z_2^3 - u^2 - 3z_2^2 + 3z_2 - 1},-\tfrac{(-1 + x)^2(2u^2u' + x^3 + 2u^2 - 3x^2 + 3x - 1)}{u(-x^3 + u^2 + 3x^2 - 3x + 1)^2}\right)$ leads to the realization 
$$x'= \frac{1-x}{u},\,y=\frac{(1-x)^2}{u^2 + (1-x)^3}.$$
Let us note that since the degree of $s$ is small, we can choose $r=s^{-1}=1+\sqrt{1-x}$ to find the same observable realization (cf. Remark~\ref{rem:nonrationalreparametrization}).
\end{example}

\section{Real Realizations}\label{sec-real}
In this section, we focus in the analysis of realizations where all coefficients are real. More precisely, a realization of the form~\eqref{eq-realization} is called real if the right hand sides $\mathbf{p},q$ are rational functions with real coefficients in the indeterminates $x_1,\ldots,x_n,u$. 
By the observation after Lemma~\ref{lem:realizibility}, it is necessary and sufficient to have $P_0 \in \RR(u)(\mathbf{x})$ and solve~\eqref{eq-param} for $\mathbf{z} \in \RR(u,\mathbf{x})^n$.

\begin{problem}
    In this section we treat the \emph{real rational realization} problem. It is formulated as follows. Let $F\in \mathbb{K}[u^{\infty}, y^{(\infty)}]$. Find a system $\Sigma$ as in \eqref{eq-realization} defined over $\mathbb{R}$ such that $F=0$ is its IO-equation. That is, find a real realization of $F$.
\end{problem}

For a given realization~\eqref{eq-realization} which is real, the corresponding parametrization $$\mathbf{P}=(q,\LD_{\mathbf{p}}(q),\ldots,\LD^n_{\mathbf{p}}(q))$$ is also real. 
So, implicitizing, one gets that the associated hypersurface is $\RR$-definable. Thus, it is sufficient to study real irreducible varieties $\Va(F)$, i.e. the irreducible $\RR$-definable ones that contain a dense set of real points and admit a (uni-)rational parametrization $\mathbf{P} \in \RR(u^{(\infty)})(\mathbf{x})^{n+1}$.

\begin{theorem}\label{thm:realrealization}
Let $F \in \RR[u^{(\infty)},y^{(\infty)}]$ be an irreducible differential polynomial of order $n$. 
There exists a real realization of $F$ if and only if $\Va(F)$ admits a real parametrization $\mathbf{P} \in \RR(u^{(\infty)})(\mathbf{x})^{n+1}$ such that $P_0 \in \RR(u)(\mathbf{x})$ and~\eqref{eq-param} is independent of derivatives of $u$.
\end{theorem}
\begin{proof}
If $\mathbf{p},q$ defines a real realization of $F$, then $(q,\LD_{\mathbf{p}}(q),\ldots,\LD^n_{\mathbf{p}}(q))$ is a real parametrization of $\Va(F)$ and, by Lemma~\ref{lem:realizibility},~\eqref{eq-param} is independent of derivatives of $u$.

Now let $\mathbf{P} \in \RR(u)(\mathbf{x}) \times \RR(u^{(\infty)})(\mathbf{x})^{n}$ be a parametrization of $\Va(F)$ and let $\mathbf{z}$ be as in~\eqref{eq-param}. 
By Lemma~\ref{lem:realizibility}, $\mathbf{x}'=\mathbf{z},y=P_0$ is a realization of $F$. Moreover, since $\mathbf{z}$ is obtained as the product of real matrices, the realization is real.
\end{proof}

Note that Theorem~\ref{thm:realrealization} holds for every subfield of $\overline{\CC(u^{(\infty)})}$ in an analogous way. 
In the following, we present the theory taking $\RR(u^{(\infty)})$ as ground field, but the reasoning is analogous if we take any real computable subfield of $\overline{\CC(u^{(\infty)})}$.

In the case of $n=1$, all real parametrizations of $\Va(F)$ can be computed algorithmically (see Lemma~\ref{lem-fieldextesionproper} and~\cite[Chapter 7]{sendra2008rational}). 
By considering the following lemma,~\cite[Algorithm 1]{pavlov2022realizing} can be used for finding real rational realizations of a given IO-equation $F \in \QQ[u,y,y']$. 
The only additional considerations are that one has to compute in step (S8) the irreducible factors {\em over $\RR$} and check whether there exists in step (S8)a a {\em real} parametrization.

\begin{lemma}\label{lem:Algorithm1RealOutput}
Assume that $F \in \QQ[u,y^{(\infty)}]$ and the parametrization computed in step (S8)a of~\cite[Algorithm 1]{pavlov2022realizing} is real. 
Then the returned realization, provided by~\cite[Algorithm 1]{pavlov2022realizing}, is real as well.
\end{lemma}
\begin{proof}
In the notation of~\cite{pavlov2022realizing}:
For a real IO-equation $F$, every rational univariate representation in step (S7) is real. 
In step (S8), it suffices to consider the real factors $r$ of $q$ irreducible over $\RR$~\cite[Lemma 7.5]{sendra2008rational}. 
Then a real parametrization $\alpha$ computed in step (S8)a leads to a real output.
\end{proof}

In the case of $n=2$, there is no algorithm known for finding real parametrizations, but real proper parametrizations can be computed. 
Based on Theorem~\ref{thm:properrealization2}, we thus can decide the existence of an observable real realization of a given IO-equation $F \in \QQ[u,y,y',y'']$ by following~\cite[Algorithm 1]{pavlov2022realizing} restricted to real proper parametrizations in step (S8)a. Note that
\begin{enumerate}
    \item The computation of a rational univariate representation does not involve field extensions. Consequently, in step (S7), $q \in \RR(z_0,z_1)[T]$.
    \item Real proper parametrizations (of irreducible factors of $q$ over $\RR$) can be computed by~\cite{schicho1998rational}.
\end{enumerate}

Let us summarize this in the following theorem.

\begin{theorem}\label{thm-realrealization2}
Let $F \in \QQ[u,y,y',y'']$ be irreducible. 
Then an observable real realization of $F$ can be computed if it exists.
\end{theorem}

For $n=2$, as commented above, real proper parametrizations of a surface can be found if they exist. 
If there does not exist a real proper parametrization, there might still exist improper real parametrizations. 
Thus, we might still find real realizations of $F \in \RR[u^{(\infty)},y,y',y'']$ even though $\Va(F)$ does not admit a proper parametrization over $\RR(u,u',u'')$, independent of the order of $F$ in $u$.

\begin{example}
Let us consider the differential polynomial $F \in \RR[u,y,y',y'']$ where
\[F = u^3 - 3u^2y'' + 3uy''^2 - y''^3 - 5u^2 + 10uy'' + y^2 + y'^2 - 5y''^2 + 4u - 4y''.\]
$F$ admits a real realization
{\small
\begin{align*}
&x_1' = \frac{(-x_1^3x_2^2 + x_1^2x_2^2 + ux_2^2 + x_1^3 - 6x_1^2x_2 + 2x_1x_2^2 + x_1^2 + u - 2x_1)(2x_1^2x_2 - 3x_1x_2^2 + 3x_1 - 4x_2)}{(3x_1^4 + 10x_1^2 + 4)(x_2^2 + 1)^2},\\
&x_2' = \frac{-12x_1^5x_2^2 + (-3x_2^4 + 30x_2^3 - 30x_2 + 3)x_1^4 + (-18x_2^4 - 12x_2^3 + 68x_2^2 - 12x_2 - 18)x_1^3}{2x_1(x_2^2 + 1)(3x_1^4 + 10x_1^2 + 4)} \\ & \qquad + \frac{((-3u + 2)x_2^4 - 36x_2^3 + 36x_2 + 3u - 2)x_1^2 - 12x_2(ux_2^2 + u + 4x_2/3)x_1 + 2ux_2^4 - 2u}{2x_1(x_2^2 + 1)(3x_1^4 + 10x_1^2 + 4)}, \\ 
&y = -\frac{((x_2^2 - 1)x_1^2 + 6x_1x_2 - 2x_2^2 + 2)x_1}{x_2^2 + 1}
\end{align*}
}
The corresponding parametrization $\mathbf{P}$ is improper. 
The surface defined by $F$ specialized at $u=1$,
$$F_s=-y''^3 + y^2 + y'^2 - 2y''^2 + 3y'',$$
has the parametrization $\mathbf{P}_s=\mathbf{P}|_{u=1}$. Since the projectivization of $F_s$ has two smooth real components, and the number of real components is a birational invariant, $\Va(F_s)$ can not be properly parametrized over the reals~\cite[Example 1]{schicho1998real}\footnote{Let us mention that in~\cite{schicho1998real} the defining equation of the cubic surface is supposed to be $F_s$.}. 
Since the specialization at $u=1$ is regular, the same holds for $F$.

Let us apply~\cite[Algorithm 1]{pavlov2022realizing} to $F$. 
In step (S8) we obtain the surface given by $$q=w^3 + 5w^2 - w^2 - z_1^2 + 4z_0.$$ 
A proper complex rational parametrization is given by
{\small
\[ \mathbf{\alpha} = \left( s, \frac{(-s^2 - 6s - 4)t^2 + 2\mathrm{i}(s + 2)^2t + s^2 + 6s + 4}{2t^2 + 2}, \frac{\mathrm{i}(s + 2)^2t^2 + (2s^2 + 12s + 8)t - \mathrm{i}(s + 2)^2}{2t^2 + 2} \right),\]}
\noindent and leads to an observable complex realization of $F$.
A real rational (improper) parametrization of $q$ can be found as well leading to the real realization above.
\end{example}

\subsection{First order IO-equations}\label{sec-real2}
For computing real realizations of $F \in \RR[u,u',y,y']$, we follow an alternative approach than the one in~\cite{pavlov2022realizing} that directly works with the parametrizations corresponding to the realizations. 
Let us note, however, that if both parametrizations used in~\cite[Algorithm 2]{pavlov2022realizing} are real, then the resulting realization is also real. 
Whether~\cite[Algorithm 2]{pavlov2022realizing} can be directly used to decide the existence of a real realization remains as an open problem.

\begin{theorem}\label{thm:realreparametrization}
Let $F \in \RR[u,u',y,y']$ be irreducible. 
Then there exists a real realization of $F$ if and only if $F$ admits a proper realization $x'=p(x,u), y=q(x,u)$ such that there exists a Lie-suitable reparametrization $s \in \CC(x)$ of the corresponding parametrization $\mathbf{P}=(q,\LD_p(q)) \in \CC(u,u')(x)^2$ with $\mathbf{P}(s) \in \RR(u,u')(x)^2$.
\end{theorem}
\begin{proof}
Assume that there exists a realization $x'=f(x,u) \in \RR(u,x), y=g(x,u) \in \RR(u,x)$ of $F$ corresponding to a real parametrization $\mathbf{Q}(x)$. 
By Theorem~\ref{thm:firstorderproper}, there exists a proper (possibly complex) realization providing a proper parametrization $\mathbf{P}(x)$ of $\Va(F)$. Now, by Theorem~\ref{thm:singletonrealizationbasis}, $\mathbf{P}(s)=\mathbf{Q}(x)$ for some $s\in \CC(x) \setminus \CC$.

Conversely, let $s \in \CC(x) \setminus \CC$ with $\mathbf{Q}(x):=\mathbf{P}(s) \in \RR(u,u')(x)^2$. 
Apply~\eqref{eq-param} to $\mathbf{Q}(x)$. 
By Proposition~\ref{prop:reparametrizationHigherOrder}, and the fact that since there are only derivatives and inversion involved and none of them require field extensions, $z \in \RR(u,x)$ and the realization $x'=z, y=Q_1$ is real.
\end{proof}

For an algorithmic way of finding reparametrizations as in Theorem~\ref{thm:realreparametrization}, we follow the works~\cite{recio1997real,andradas2009simplification}.
For this purpose, let us introduce analytic functions and present their relation to the problem of finding real Lie-suitable reparametrizations.

\begin{definition}
A rational function $r(x,z)\in \CC(\mathbf{u})(x,z)$ is called \textit{analytic} if there exists $g(w)\in \CC(\mathbf{u})(w)$ such that 
\[ g(x+\mathrm{i}\,z)=r(x,z) .\]
In the affirmative case, $g$ is called the \textit{generator} of $r$.
\end{definition}

\begin{remark}
Every $r(\mathbf{x})\in \CC(\mathbf{u})(\mathbf{x})$, analytic or not, can be written as $g=U(\mathbf{x})+\mathrm{i}\,V(\mathbf{x})$ with $U,V\in \RR(\mathbf{x})$. We call $U,V$ the \textit{(real and imaginary) components} of $r$.

For a rational function $r(\mathbf{x})\in \CC(\mathbf{u})(\mathbf{x})$ with components $U,V$, we write $\overline{r}(\mathbf{x})$ for the conjugation $\overline{r}(\mathbf{x}):=U(\mathbf{x})-\mathrm{i}\,V(\mathbf{x})$.
\end{remark}

The next lemmas generalize~\cite[Lemma 2.1., 2.2.]{recio1997real}.
Given $p(x,z)\in \CC(\mathbf{u})[x,z]$ and $\mathbf{u}^0 \in \Omega$ in some set $\Omega$, we will use the notation $p(\mathbf{u}^0;x,z)$ for evaluating the variables $\mathbf{u}$ at $\mathbf{u}^0$ whenever it is well-defined.

\begin{lemma}\label{lem-lemma21InRS97}
Let $p(x,z)\in \CC(\mathbf{u})[x,z]\setminus \{0\}$ be an analytic polynomial and $U,V\in \RR[x,y]$ its real and imaginary part. Then, $\gcd(U,V)=1$.
\end{lemma}
\begin{proof}
If $p$ is constant, the result is trivial. 
Let $p(x,z)$ be non-constant and let $g(w)\in \CC(\mathbf{u})[w]$ be the polynomial generator. Let $G:=\gcd(U,V)$ and let $U^*$ and $V^*$ be the corresponding cofactors, i.e. $U=U^*G$ and $V=V^*G$. We consider the following non-empty Zariski open subsets of $\RR^{\#(\mathbf{u})}$. 
Let $\Omega_1\subset \RR^{\#(\mathbf{u})}$ be such that, for $\mathbf{u}^0\in \Omega_1$, the polynomials $p(\mathbf{u}^0;x,z), U(\mathbf{u}^0;x,z), V(\mathbf{u}^0;x,z)$ are well-defined and 
\begin{enumerate}
\item $\deg_w(g(w))=\deg_w(g(\mathbf{u}^0;w)),$
\item 
$\deg_{\{x,z\}}(U(x,z))=\deg_{\{x,z\}}(U(\mathbf{u}^0;x,z)),$
\item $\deg_{\{x,z\}}(V(x,z))=\deg_{\{x,z\}}(V(\mathbf{u}^0;x,z))$. 
\end{enumerate}
$\Omega_1$ can be constructed by taking the lcm of all denominators of all non-zero coefficients of $g$ w.r.t. $w$ and of $U,V$ w.r.t. $\{x,z\}$. In addition, one has also to require that the leading coefficient of $g$ w.r.t. $w$ does not vanish and that at least one non-zero coefficient of each of the corresponding leading terms of $U,V$, seen as polynomials in $\RR[x,z]$, does not vanish.

We observe now that for $\mathbf{u}^0\in \Omega_1$ it holds that $p(\mathbf{u}^0;x,z)$ is analytic, generated by $g(\mathbf{u}^0;w)$, with real and imaginary part $U(\mathbf{u}^0;x,z)$ and $V(\mathbf{u}^0;x,z)$, respectively. Moreover, by~\cite[Lemma 2.1]{recio1997real}, we get that $\gcd(U(\mathbf{u}^0;x,z),V(\mathbf{u}^0;x,z))=1$.

Let $\Omega_2\subset \Omega_1$ be such that for $\mathbf{u}^0\in \Omega_2$ it holds that $U^*(\mathbf{u}^0;x,z),V^*(\mathbf{u}^0;x,z),G(\mathbf{u}^0;x,z),$ are well-defined, and 
\begin{itemize}
\item[4.] 
$\deg_{\{x,z\}}(U^*(x,z))=\deg_{\{x,z\}}(U^*(\mathbf{u}^0;x,z)),$
\item[5.] $\deg_{\{x,z\}}(V^*(x,z))=\deg_{\{x,z\}}(V^*(\mathbf{u}^0;x,z))$. 
\end{itemize}
$\Omega_2$ can be constructed following similar comments as in the construction of $\Omega_1$. Clearly $\Omega_2\neq \emptyset$
By (2),(3),(4),(5), we get that $\deg_{\{x,z\}}(G(x,z))=\deg_{\{x,z\}}(G(\mathbf{u}^0;x,z))$. Furthermore, $G(\mathbf{u}^0;x,z)$ divides $U(\mathbf{u}^0;x,z)$ and $V(\mathbf{u}^0;x,z)$. So, $G(\mathbf{u}^0;x,z)$ divides $\gcd(U(\mathbf{u}^0;x,z),V(\mathbf{u}^0;x,z))=1$, and hence $\deg_{\{x,z\}}(G(x,z))=\deg_{\{x,z\}}(G(\mathbf{u}^0;x,z))=0$. Thus, $\gcd(U,V)=1$.
\end{proof}

Let $\mathbf{P}=\left(\tfrac{f_1}{g_1},\tfrac{f_2}{g_2}\right) \in \CC(\mathbf{u})(x)^2$ be a parametrization such that $\gcd(f_1,g_1)=\gcd(f_2,g_2)=1$ where the $\gcd$ is taken over $\CC(\mathbf{u})[x]$. 
We consider the formal substitution
$\mathbf{P}^*(x,z):=\mathbf{P}(x+\mathrm{i}\,z)$ that we express as
\begin{equation}\label{eq-P}
\mathbf{P}^*(x,z)=\left( \dfrac{U_1(x,z)+\mathrm{i}\,V_1(x,z)}{W_1(x,z)^2}, \dfrac{U_2(x,z)+\mathrm{i}\,V_2(x,z)}{W_2(x,z)^2}\right)\in \CC(\mathbf{u})(x,z)^2
\end{equation}
where $U_i,V_i$ are the real and imaginary parts of $f_i(x+\mathrm{i}\,y) \overline{g_i}(x-\mathrm{i}\,y)$, $A_i,B_i$ are the real and imaginary parts of $g_i(x+\mathrm{i}\,y)$, and $W_i:= A_{i}^{2}+B_{i}^{2}$.

\begin{remark}\label{rem-conformal}
Let $s\in \CC(x)$ have degree at least one in $x$, and let $s_1,s_2 \in \RR(x)$ be the real and imaginary parts of $s$. Clearly $(s_1,s_2)\not\in \RR^2$ since otherwise $s\in \CC$. So $(s_1,s_2)$ is a parametrization of a real curve. If $\deg_x(s)=1$, i.e. $s$ is a M\"obius transformation, the curve parametrized by $(s_1,s_2)$ is either a real line or a real circle since it is the image of a real line (namely $\RR$) under a conformal map.
\end{remark}

\begin{theorem}\label{theorem-characterization}
Let $\mathbf{P} \in \CC(\mathbf{u})(x)^2$ be a proper parametrization of $F \in \RR(\mathbf{u})[y_0,y_1]$ and let $\mathbf{P}^*$ be as in~\eqref{eq-P}. 
Then the following statements are equivalent.
\begin{enumerate}
\item There exists a reparametrization given by $s \in \CC(x)$ such that $\mathbf{P}(s) \in \RR(\mathbf{u})(x)^2$.
\item $V:=\gcd(V_1,V_2)$ has a factor in $\RR[x,z]$ that defines a real rational curve.
\end{enumerate}
In the affirmative case, if $(s_1,s_2) \in \RR(x)^2$ is a parametrization of the real rational curve stated in~(2), then $s=s_1+\mathrm{i}\,s_2$ fulfills~(1). 
Moreover, it holds that $\deg_x(s)=1$ if and only if $V$ defines a real line or a real circle.
\end{theorem}
\begin{proof} 
Assume that $\mathbf{P}$ fulfills (1) for some reparametrization $s\in \CC(x)$. Let $s$ be expressed as $s_1+\mathrm{i}\,s_2$ where $s_1,s_2$ are the real and imaginary parts of $s$, respectively. By Remark~\ref{rem-conformal}, $(s_1,s_2) \in \RR(x)^2 \setminus \RR^2$. Let $W_i(s_1,s_2)$ with $W_i=(A_i+\mathrm{i}\,B_i)(A_i- \mathrm{i}\,B_i)$ be as in~\eqref{eq-P}. If $W_i(s_1,s_2)=0$, since $s_1,s_2$ are real, then $A_i(s_1,s_2)=B_i(s_1,s_2)=0$ and this implies that $\gcd(A_i,B_i)\neq 1$ which is a contradiction to Lemma~\ref{lem-lemma21InRS97}.
Therefore, $W_i(s_1,s_2)$ is non-zero and $\mathbf{P}^*(s_1,s_2)=\mathbf{P}(s)\in \RR(\mathbf{u})(x)^2$ is well-defined. Thus, since $s_1,s_2\in \RR(x)$ we have that $V_1(s_1,s_2)=V_2(s_1,s_2)=0$. So, $(s_1,s_2)$ parametrizes the curve defined by one factor $V^*$ of $V$. Since $(s_1,s_2)\in \RR(x)^2$ we have that $V^*\in \RR[x,z]$ defines a real rational curve. By Remark~\ref{rem-conformal}, if $s$ is a M\"obius-transformation, $\Va(V^*)$ is either a real line or a real circle.

For the converse direction, let $(a,b) \in \RR(x)^2$ be a proper parametrization of a factor $V^*$ of $V$ defining a real rational curve. Reasoning as above, we have that $W_i(a,b)\neq 0$. Therefore, $\mathbf{P}(a+\mathrm{i}\,b)\in \RR(\mathbf{u})(x)^2$. We observe that $a+\mathrm{i}\,b$ is a rational function of positive degree and hence, not both components of $\mathbf{P}(a+\mathrm{i}\,b)$ can be constant. Thus, $\mathbf{P}(a+\mathrm{i}\,b)$ is indeed a parametrization. Moreover, in the proof of~\cite[Theorem 3.2.]{recio1997real} it is shown that if $V^*$ defines a real line or a real circle, then $s=a+\mathrm{i}\,b$ is a M\"obius transformation.
\end{proof}

\begin{corollary}\label{thm:realproperrealization}
Let $F \in \RR[u,u',y,y']$ be irreducible with an observable realization $x'=p(u,x), y=q(u,x)$. 
Then there exists a real realization of $F$ if and only if $V$, as in Theorem~\ref{theorem-characterization}, has a factor $V^* \in \RR[x,z]$ defining a real rational curve. 
Moreover, there exists an observable real realization if and only if $V^*$ defines a real line or a real circle.
\end{corollary}
\begin{proof}
Theorem~\ref{thm:firstorderproper} implies that $F$ is observably realizable as stated. Let $\mathbf{P}=(q,\LD_p(q)) \in \CC(u,u')(x)^2$ be the corresponding parametrization. By Theorem~\ref{thm:realreparametrization}, there exists a real realization if and only if it can be obtained by a Lie-suitable reparametrization of $\mathbf{P}$. Then, the statement follows by Theorem~\ref{theorem-characterization}.
\end{proof}

Note that if an observable real realization exists, we can find every other real realization by a real Lie-suitable reparametrization.

\begin{algorithm}[H]
\caption{RealRealization}
\label{alg-realrealization}
\begin{algorithmic}[1]
    \REQUIRE An irreducible polynomials $F \in \QQ[u,u',y,y']$.
    \ENSURE A real realization of $F$ if it exists.
    \STATE Decide whether $F$ is realizable and, in the affirmative case, compute an observable realization $x'=p(x,u), y=q(x,u)$ of $F$ by Algorithm~\ref{alg-properrealization}.
    \STATE Compute $V=\gcd(V_1,V_2)$ as in Theorem~\ref{theorem-characterization} from the corresponding parametrization $\mathbf{P}=(q,\LD_{p}(q))$.
    \STATE If a factor of $V$ defines a real rational curve $\mathcal{C}$, then compute a real proper parametrization $(s_1,s_2) \in \RR(x)^2$ of $\mathcal{C}$ and set $s_i=s_1+\mathrm{i}\,s_2$; otherwise stop.
    \STATE Output the real realization $x'=p(u,s_i)/ \partial_xs_i, \, y=q(u,s_i)$.
\end{algorithmic}
\end{algorithm}

\begin{theorem}\label{thm-allrealrealizatons}
Algorithm~\ref{alg-realrealization} is correct.
\end{theorem}
\begin{proof}
Corollary~\ref{thm:realproperrealization} together with Theorem~\ref{thm-Alg1correct} imply correctness. 
Termination follows from the termination of Algorithm~\ref{alg-properrealization} and the termination of deciding whether a factor of $V$ is real rational~\cite[Section 7]{sendra2008rational}.
\end{proof}

\begin{example}
Let us consider the first-order IO-equation
\begin{align*}
F = \, & 9(u - 1)^2y^4 + (-12u^2-24u+36)y^3+(22u^2+128y'^2-12u + 54)y^2 + (-12u^2-24u + 36)y \\ & + 9u^2 + 128y'^2 - 18u + 9 \, = 0.
\end{align*}
A (complex) realization can be found by applying~\cite[Algorithm 2]{pavlov2022realizing} with
\[ x'=\frac{\mathrm{i}\,(x^2-2x-1)(ux^4 - 6x^2 + u)}{8(x^2 + 1)^2}, \, y=\frac{-x^2-2x+1}{x^2-2x-1}. \]
The corresponding parametrization is
\[ \mathbf{P} = \left(\frac{-x^2 - 2x + 1}{x^2 - 2x - 1}, \frac{\mathrm{i}\,(ux^4 - 6x^2 + u)}{2x^4 - 4x^3 - 4x - 2}\right).\]
Evaluating the defining polynomial $F$ at $u=1$, we obtain the irreducible polynomial $$F_s = 2y^2y'^2 + y^2 + 2y'^2.$$
The corresponding curve $\Va(F_s)$ can be rationally parametrized by the evaluation of $\mathbf{P}$ at $u=1$, but there exists no real parametrization~\cite[page 252]{recio1997real} and thus no real realization of $F_s$. 
Alternatively, by following Algorithm~\ref{alg-realrealization}, we obtain $V=x^2+z^2+1$ which defines a non-real curve.
\end{example}

\subsubsection*{Acknowledgements}
The authors would like to thank Gleb Pogudin and Josef Schicho for useful discussions. 
First and third author partially supported by the grant PID2020-113192GB-I00/AEI/10.13039/501100011033 (Mathematical Visualization: Foundations, Algorithms and Applications) from the Spanish State Research Agency (Ministerio de Ciencia, Innovación y Universidades). 

\appendix

\bibliographystyle{plain}

\end{document}